\documentclass[12pt,draftclsnofoot, onecolumn]{IEEEtran}
\usepackage{cite}
\usepackage{amsmath}
\usepackage{amsthm}
\usepackage{amssymb}
\usepackage{amsfonts}
\usepackage{mathrsfs}
\usepackage[ruled,vlined,linesnumbered]{algorithm2e}
\usepackage{multirow}
\usepackage{array}
\usepackage{graphicx}
\usepackage{fixltx2e}
\usepackage{dblfloatfix}
\usepackage{bm}
\usepackage{subcaption}
\usepackage{color}
\definecolor{color1}{RGB}{128, 0, 0}
\definecolor{color2}{RGB}{0, 0, 128}
\usepackage[colorlinks, linkcolor=color1, anchorcolor=blue, citecolor=color1, urlcolor = color2]{hyperref}

\usepackage{url}
\usepackage{autobreak}
\usepackage{tabularx}
\newtheorem{lemma}{Lemma}

\newtheorem{proposition}{Proposition}

\newtheorem{remark}{Remark}
\usepackage{pifont}

\DeclareMathOperator*{\argmin}{arg\,min}

\allowdisplaybreaks

\begin{document}

\title{Over-the-Air Computation via Reconfigurable Intelligent Surface}

\author{Wenzhi Fang, \textit{Student Member}, \textit{IEEE},  Yuning Jiang, \textit{Member}, \textit{IEEE}, \\ Yuanming Shi, \textit{Senior Member}, \textit{IEEE}, Yong Zhou, \textit{Member}, \textit{IEEE}, \\Wei Chen, \textit{Senior Member},  \textit{IEEE},  and Khaled B. Letaief, \textit{Fellow}, \textit{IEEE}
        
\thanks{This paper was presented in part at the \textit{Proc. of IEEE Globecom}, Waikoloa, HI, Dec. 2019 \cite{jiang2019}. }
\thanks{
W. Fang, Y. Shi, and Y. Zhou are with the School of Information Science and Technology, ShanghaiTech University, Shanghai 201210, China (e-mail: \{fangwzh1, shiym, zhouyong\}@shanghaitech.edu.cn). Y. Jiang is with Automatic Control Laboratory, EPFL, Switzerland (e-mail: yuning.jiang@epfl.ch). W. Chen is with the Department of Electronic Engineering, Tsinghua University, Beijing 100084, China (e-mail: wchen@tsinghua.edu.cn). K. B. Letaief is with the Department of Electronic and Computer Engineering, Hong Kong University of Science and Technology, Hong Kong (e-mail: eekhaled@ust.hk).}% He is also with Peng Cheng Laboratory, Shenzhen, China.} 
}
%\author{Wenzhi~Fang,~\IEEEmembership{Student Member,~IEEE,} 
%       Yuning~Jiang,~\IEEEmembership{Member,~IEEE,}
%       Yuanming~Shi,~\IEEEmembership{Senior Member,~IEEE,}and~Yong~Zhou,~\IEEEmembership{Member,~IEEE,}}% <-this % stops a space
%\thanks{M. Shell was with the Department
%of Electrical and Computer Engineering, Georgia Institute of Technology, Atlanta,

%{Shell \MakeLowercase{\textit{et al.}}: An Efficient Algorithm for Over the Air Computation Aided by the Intelligent Reflecting Surface}

\maketitle

\vspace{-9mm}
\begin{abstract}
Over-the-air computation (AirComp) is a disruptive technique for fast wireless data aggregation in Internet of Things (IoT) networks via exploiting the waveform superposition property of multiple-access channels. 
However, the performance of AirComp is bottlenecked by the worst channel condition among all links between the IoT devices and the access point. 
In this paper, a reconfigurable intelligent surface (RIS) assisted AirComp system is proposed to boost the received signal power and thus mitigate the performance bottleneck by reconfiguring the propagation channels. 
With an objective to minimize the AirComp distortion, we propose a joint design of AirComp transceivers and RIS phase-shifts, which however turns out to be a highly intractable non-convex programming problem. 
To this end, we develop a novel alternating minimization framework in conjunction with the successive convex approximation technique, which is proved to converge monotonically. 
To reduce the computational complexity, we transform the subproblem in each alternation as a smooth convex-concave saddle point problem, which is then tackled by proposing a Mirror-Prox method that only involves a sequence of closed-form updates. 
Simulations show that the computation time of the proposed algorithm can be two orders of magnitude smaller than that of the state-of-the-art algorithms, while achieving a similar distortion performance. 
%By exploiting the benign structure of the resulting subproblems, we further propose a Mirror-Prox method with low computational complexity for optimizing large-scale RIS-AirComp systems. 
%Simulation results demonstrate the algorithmic advantages of the proposed algorithm and performance gains of the presented RIS-AirComp system.
\end{abstract}


\newpage
\section{Introduction}
Driven by the increasing advancement of wireless communication technologies and the decreasing manufacturing costs, Internet of Things (IoT) is expected to support ubiquitous connectivity and automatic transmission for billions of devices equipped with sensing and communication capabilities \cite{IOT}. 
With limited spectrum resources, it is generally challenging to achieve efficient wireless data aggregation over a large volume of IoT devices, which is critical for unleashing the potential of the distributed sensory data. 
The conventional ``transmit-then-compute" approach requires an access point (AP) to successfully receive the data from each IoT device and then compute a specific function (e.g., arithmetic mean) of the received data. 
This is, however, not spectrum-efficient, especially when the number of IoT devices is large. 
Fortunately, over-the-air computation (AirComp), which integrates the communication and computation processes, has the potential to achieve ultra-fast wireless data aggregation in IoT networks. 
This is accomplished by enabling the concurrent data transmissions from all IoT devices over the same radio channel and exploiting the waveform superposition property of multiple-access channels (MACs) at the AP \cite{sum1,sum4}, yielding a revolutionary paradigm of ``compute when communicate". 

The study of AirComp can be traced back to the seminal work \cite{info}, which showed that a fast function computation can be achieved by enabling concurrent analog transmissions. 
There is a growing body of studies concentrated on the transceiver design for AirComp to enable efficient wireless data aggregation \cite{optimal_liu, optimal_huang, chen2018uniform, multi_function, Muti_modal}.
In particular, the authors in \cite{optimal_liu, optimal_huang} proposed optimal transmit power control strategies for AirComp in single-input single-output (SISO) wireless networks with energy-constrained IoT devices. 
The authors in \cite{chen2018uniform} studied AirComp in multiple-input single-output (MISO) wireless networks, where a novel uniform-forcing transmit design was proposed to compensate the non-uniform channel fading among IoT devices. 
By integrating multiple-input multiple-output (MIMO) with AirComp, the authors in \cite{multi_function} and \cite{Muti_modal} investigated the transceiver design for multi-function computation and multi-modal sensing, respectively.
%In particular, the work \cite{Muti_modal} exploited the geometric structure of the receive beamformer on a Grassman manifold.
Meanwhile, the authors in \cite{blind} proposed a blind MIMO AirComp scheme to reduce the signaling overhead for channel state information (CSI) estimation. 
By calibrating the transmission timing of each IoT device, the synchronization issue of AirComp can be addressed. 
In case of non-strict synchronization, the authors in \cite{shao2021federated} proposed an efficient matched filtering and sampling scheme to facilitate misaligned AirComp. 
More recently, the authors in \cite{yangkai, FL_aircomp_Amiri, zhangpeng} exploited the advantages of AirComp to develop a fast model aggregation scheme to accelerate the convergence of federated learning.
According to the aforementioned studies, AirComp requires the magnitudes of the signals to be aligned at the AP; thus, the performance of AirComp is bottlenecked by the worst channel between the IoT devices and the AP.

Reconfigurable intelligent surface (RIS) is an emerging technology, which has recently been proposed to tackle unfavorable channel conditions by reconfiguring the radio propagation environment \cite{liangyichang,zhang,wu2019intelligent, ReconfigurableIRS,xiaojunRIS,fu2020reconfigurable,xia2021reconfigurable,he2020reconfigurable,mmWave,UAV,Energy_harvesting}. 
An RIS is a man-made flat surface composed of many passive reflecting elements, each of which can independently shift the phase of the impinging waves in a controllable way \cite{liangyichang,zhang}, thereby constructing a favorable wireless radio propagation environment. 
The authors in \cite{wu2019intelligent,ReconfigurableIRS,xiaojunRIS} demonstrated that RIS has the ability to significantly enhance the energy efficiency and spectral efficiency of wireless networks. 
Owing to the aforementioned features, RIS has been integrated with various wireless technologies, e.g., non-orthogonal multiple access (NOMA) \cite{fu2020reconfigurable}, massive IoT device connectivity \cite{xia2021reconfigurable}, massive MIMO \cite{he2020reconfigurable}, millimeter-wave communications \cite{mmWave}, unmanned aerial vehicle communications \cite{UAV}, and wireless power transfer \cite{Energy_harvesting}, to further enhance the network performance and promote emerging applications. 

To leverage the advantages of RIS for enhancing the quality of the worst channel between the IoT devices and the AP and in turn mitigating the performance bottleneck of AirComp, the authors in \cite{jiang2019, zhibin} proposed an RIS-assisted AirComp system. 
In these two studies, the alternating semi-definite relaxation (SDR) algorithm and the alternating difference-of-convex (DC) algorithm were proposed to jointly optimize the receive beamforming vector at the AP and the phase-shift matrix at the RIS. 
However, both algorithms suffer from high computational complexity as they need to iteratively solve the semi-definite programming (SDP) problems.
Moreover, the optimization of the phase-shift matrix in \cite{jiang2019, zhibin} involves feasibility detection that cannot be accurately tackled by the SDR and DC algorithms. 
Hence, both algorithms are not guaranteed to achieve monotonic convergence. 
%The proposed alternating DC algorithm in \cite{zhibin} has a convergence guarantee, however, which is not scalable since that need to solve a series of SDP problem in each alternating iteration.
This motivates us to develop a computationally efficient algorithm with convergence guarantee to achieve efficient wireless data aggregation in RIS-assisted AirComp systems.

\subsection{Contributions}
In this paper, we consider an RIS-assisted IoT network, where a multi-antenna AP aggregates the sensory data from multiple IoT devices by using AirComp with the assistance of an RIS. 
We evaluate the performance of AirComp in terms of the computation distortion, which is measured by the mean-squared-error (MSE).
The main objective is to develop a low-complexity algorithm with convergence guarantee to minimize the MSE of RIS-assisted AirComp systems. 
To this end, it is necessary to jointly optimize the transmit scalars at the IoT devices, the receive beamforming vector and the denoising factor at the AP, and the phase-shifts at the RIS. 
The main contributions of this paper are summarized as follows:
\begin{itemize}
\item 
We formulate an MSE minimization problem for RIS-assisted AirComp systems, where the transmit scalars at the IoT devices, the receive beamforming vector and the denoising factor at the AP, and the phase-shift matrix at the RIS are jointly optimized. 
To tackle the scalability and convergence issues of existing studies, we transform the original problem into an equivalent min-max optimization problem. 
It turns out to be a highly intractable non-convex optimization problem due to the non-convex objective function and the coupled optimization variables. 
\item 
We propose an alternating minimization (AlterMin) framework to alternately optimize the receive beamforming vector and the phase-shift vector.
For each optimization problem in the alternating procedure, we iteratively construct a convex surrogate for the non-convex objective function by utilizing the successive convex approximation (SCA) technique.
We prove that the proposed framework is guaranteed to converge, which is a key difference from the existing studies. 
As the subproblem in each SCA iteration is a non-smooth convex problem, the conventional algorithms suffer from high computational complexity. 
\item 
As the objective function involves the pointwise maximum of affine functions, we equivalently transform the resulting non-smooth convex problem into a smooth convex-concave saddle point problem by using the primal-dual transformation.
Subsequently, we adopt the Mirror-Prox method to solve the aforementioned saddle point problem and derive a closed-form expression for each update. 
As a result, the proposed Mirror-Prox based AlterMin SCA algorithm enjoys a very low computational complexity. 

%Therefore, the resulting problem in each SCA iteration can be solved by a Mirror-Prox method with a fast convergence rate and low iteration cost. 
%The proposed Mirror-Prox based AlterMin SCA algorithm thus can be applied for optimizing the large-scale RIS-AirComp system.
\end{itemize}

We conduct extensive simulations to demonstrate the monotonic convergence and the superior performance of the proposed algorithm for RIS-assisted AirComp systems. 
Results will show that the computational time of the proposed algorithm can be two orders of magnitude smaller than that of the alternating SDR and alternating DC algorithms \cite{jiang2019, zhibin}, while achieving a similar data aggregation distortion performance in terms of MSE as these state-of-the-art algorithms. 
Moreover, the performance gain of the proposed algorithm in terms of the computation time increases when the dimensions of the system parameters (e.g., number of IoT devices, AP's antennas, and RIS's reflecting elements) become larger.

\subsection{Organization and Notations}
The reminder of this paper is organized as follows.
Section II describes the system model and the problem formulation. 
Section III presents an AlterMin SCA framework for solving the formulated problem. 
Section IV provides the Mirror-Prox method to solve the subproblems in each SCA iteration. 
The performance of the proposed algorithm is illustrated in Section V. 
Finally, we conclude this paper in Section VI.

\emph{Notations:}
 %The notations are given as follows. 
 Matrices, vectors, and scalars are denoted by boldface upper-case, boldface lower-case, and lower-case letters, respectively.
  $(\bm \cdot)^{\sf{H}}$ and $(\bm \cdot)^{\sf{T}}$ stand for conjugate transpose and transpose of a matrix or a vector, respectively.
 $\|\bm \cdot\|_1$, $\|\bm \cdot\|$, and $\|\bm \cdot\|_{\infty} $ denote the $l_1$, $l_2$, and $l_{\infty}$ norm operators, respectively.
 $\Re[\bm\cdot]$ and $\Im[\bm \cdot]$ represent the real and imaginary parts of a complex matrix, vector, or scalar, respectively.
 %$\text{diag}\{\bm a\}$ represent a diagnal matrix that the diagnal elements is $\bm a$.
 $\mathbb{E}\left[\bm \cdot\right]$ denotes the expectation of a random variable.

%We formulate a computation error of AirComp minimization problem with respect to the denoising factor at the AP, transmit scalars at the IoT sensors, receive beamforming at the AP, and phase-shift vector at the RIS.

%
%AirComp is a promising tech Wireless 
%
%To support these   
%
%
%It is anticipated

%\newpage

\section{System Model and Problem Formulation}

In this section, we first present the system model of an RIS-assisted AirComp system and formulate an AirComp distortion minimization problem that requires the joint optimization of the transmit scalars at the IoT devices, the receive beamforming vector and the denoising factor at the AP, and the phase-shift matrix at the RIS.
Subsequently, we discuss the limitations of the existing methods, which motivate us to reformulate an equivalent min-max optimization problem.

\begin{figure}[t]
\centering{\includegraphics[width=0.5\textwidth]{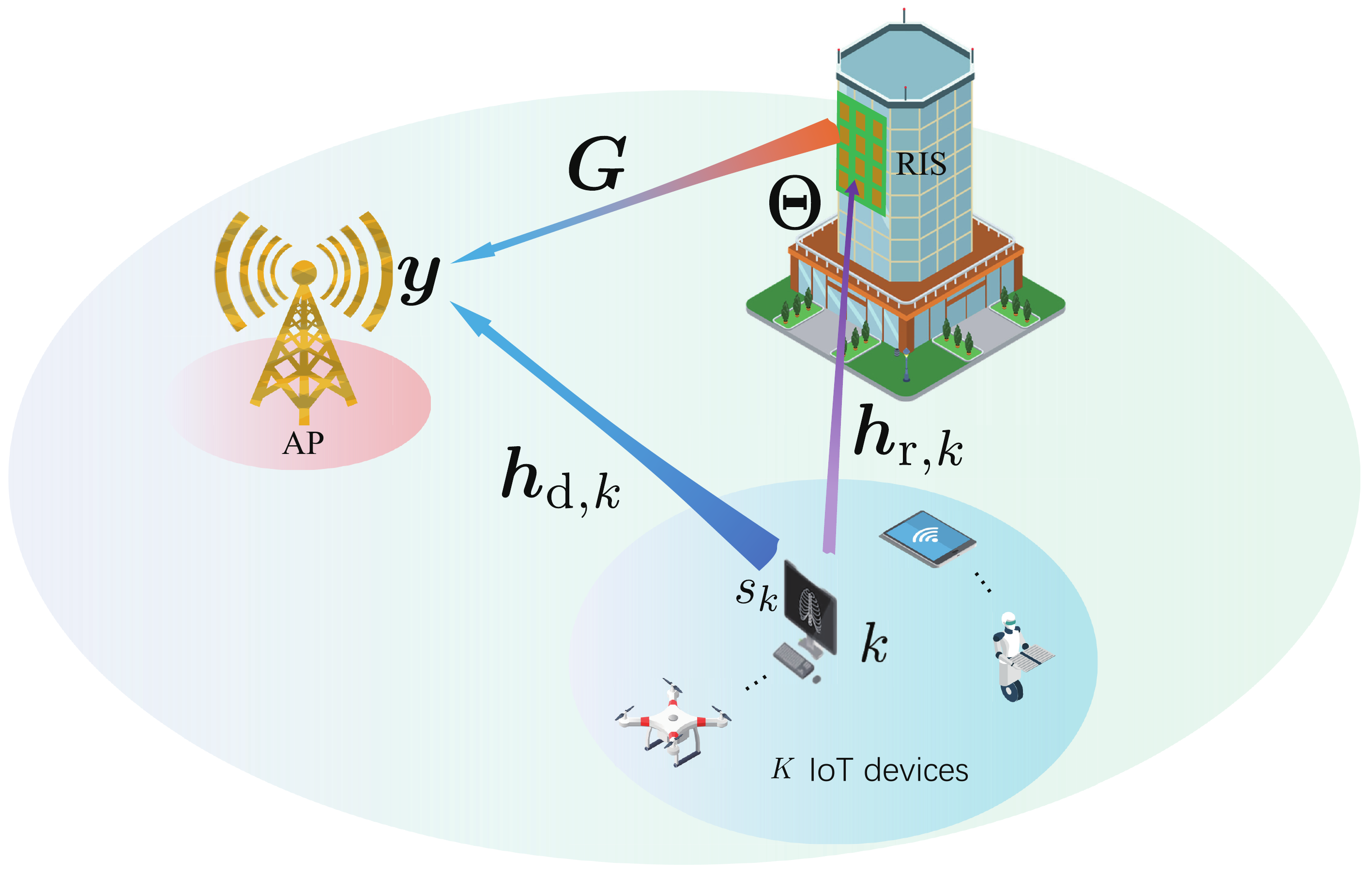}
\caption{Illustration of an RIS-assisted IoT network, where an AP aggregates the data from multiple IoT devices using AirComp. }% with the assistance of an RIS.}
\label{sysmodel}}
\vspace{-8mm}
\end{figure} 

\subsection{System Model}
We consider the uplink transmission of an RIS-assisted single-cell IoT network consisting of $K$ single-antenna IoT devices, an AP with $M$ antennas, and an RIS equipped with $N$ passive reflecting elements, as shown in Fig. \ref{sysmodel}. We denote $\mathcal{K}=\{1,2,\ldots, K\}$ as the index set of IoT  devices.
We consider the scenario that the AP, as a fusion center, is interested in receiving an aggregation (e.g., geometric mean, arithmetic mean) of the sensory data (e.g., temperature, humidity) from all IoT devices, rather than the individual data from each IoT device \cite{multi_function}.
This process is generally referred to as wireless data aggregation in IoT networks.
%For wireless data aggregation, the AP served as a fusion center is interested in receiving the average of the monitored data (e.g., temperature, humidity) at IoT devices \cite{multi_function}.
%For federated learning, the AP served as a parameter server needs to receive the weighted average of the local model parameters from the distributed edge devices rather than the local model parameters of each device to update the global model.
By integrating the computation and communication processes via exploiting the waveform superposition property of MACs, AirComp is adopted in this paper to achieve ultra-fast data aggregation by enabling concurrent transmissions from multiple IoT devices. 
We denote $z_k\in\mathbb{C}$ as the representative information-bearing data at IoT device $k$.
Before transmission, IoT device $k$ normalizes data $z_k$ as information symbol $s_k \in \mathbb{C}$.
Without loss of generality, we assume that $\{s_k\}_{k=1}^{K}$ have zero mean and unit power,  and are independent of each other, i.e., $\mathbb{E}[s_{k}] = 0$, $\mathbb{E}[s_{k} s_{k}^{\sf *}] = 1,$ and $\mathbb{E}[s_{k} s_{i}^{*}] = 0, ~ \forall ~   k \neq i$ \cite{optimal_huang,optimal_liu}.
The target function that the AP aims to recover is the summation of the data from all IoT devices, i.e., 
\begin{align}
        g=\sum_{k\in\mathcal{K}}s_k.
\end{align}
Based on the principle of AirComp, we assume that all IoT devices are synchronized and transmitted concurrently to the AP \cite{multi_function}.
With universal frequency reuse, the signal received at the AP from all IoT devices is given by
\begin{align}
\bm{y}=\sum_{k\in\mathcal{K}}(\bm h_{\mathrm{d},k}+\bm G \bm\Theta \bm  h_{\mathrm{r},k} ){w}_ks_k+\bm{n},
\end{align}
where $w_k\in\mathbb{C}$ denotes the transmit scalar of IoT device $k$, $\bm{\Theta} = \text{diag}\left\{e^{j\theta_1},e^{j\theta_2},\ldots,e^{j\theta_N}\right\}$ denotes the 
 diagonal phase-shift matrix of RIS with $0 \leq \theta_i < 2\pi, ~ \forall ~i$, 
 % ~ \text{and} ~ |\theta_n| = 1,~\forall ~  n=1,\ldots,N$, 
  $ \bm  h_{\mathrm{d},k}\in\mathbb{C}^{M\times1}, \bm G\in\mathbb{C}^{M\times N}$, and $\bm  h_{\mathrm{r},k}\in\mathbb{C}^{N\times 1}$ denote the channel coefficients of the links from device $k$ to the AP, from the RIS to the AP, and from device $k$ to the RIS, respectively,
 and $ \bm n\sim\mathcal{CN}(0, \sigma^2\bm I_M) $ is the additive white Gaussian noise (AWGN) with zero mean and variance $\sigma^2$. %Consider the power constraint at IoT devices, i.e. $|w_k|^2 \leq P,  k\in\mathcal{K}$.
Each device has a maximum transmit power, denoted as $P$. 
Hence, we have $|w_k|^2 \leq P, \forall \, k\in\mathcal{K}$.
As various effective channel estimation methods have been proposed for RIS-assisted wireless networks \cite{shuguang,xiaojun}, 
we assume that perfect CSI is available in our work, as in \cite{wu2019intelligent, ReconfigurableIRS,xiaojunRIS,fu2020reconfigurable,xia2021reconfigurable,he2020reconfigurable,mmWave,UAV,Energy_harvesting,jiang2019,zhibin}.
Besides, we consider block-fading channels, where the channel gain of each link remains invariant within one time slot but varies independently across different time slots. 
The estimated function at the AP is given by \cite{yangkai}
%\begin{align}
%\hat{g}&={\bm{a}^{\sf{H}}\bm{y}}\! \nonumber \\ 
%&={\bm{a}}^{\sf{H}}\sum_{k\in\mathcal{K}}(\bm  h_{\mathrm{d},k} \!+\!\bm G \bm\Theta \bm  h_{\mathrm{r},k} ){w}_ks_k \!+\!\bm{a}^{\sf{H}}\bm{n},
%\end{align} 
%where $\bm{a}\in\mathbb{C}^M$ denotes the receive beamforming vector.
%We adopt MSE to measure the distortion between the estimated function and the desired target one, which quantifies the AirComp performance and can be expressed as \cite{jiang2019}
%\begin{align}
%\label{MSE_ori}
%&{\sf{MSE}}(\hat{g}, g) = \mathbb{E}\left(|\hat{g}-g|^2\right)\nonumber \\
%=& \!\! \sum_{k\in\mathcal{K}}\left|\bm{a}^{\sf{H}}(\bm  h_{\mathrm{d},k}+\bm G \bm\Theta \bm  h_{\mathrm{r},k} ){w}_k -1\right|^2 \!+\!\sigma^2\|\bm{a}\|^2. 
%\end{align}
%The MSE given in \eqref{MSE_ori} depends on the transmit scalar $w_k$, the receive beamforming vector $\bm a$, and the phase-shift matrix $\bm \Theta$.
\begin{align}
\hat{g}={1\over{\sqrt \eta}}{\bm{m}^{\sf{H}}\bm{y}} ={1\over{\sqrt \eta}}{\bm{m}}^{\sf{H}}\sum_{k\in\mathcal{K}}(\bm  h_{\mathrm{d},k} +\bm G \bm\Theta \bm  h_{\mathrm{r},k} ){w}_ks_k +{1\over{\sqrt \eta}}\bm{m}^{\sf{H}}\bm{n},
\end{align} 
where $\bm{m}\in\mathbb{C}^M$ and $\eta$ denote the receive beamforming vector and the denoising factor at the AP, respectively.  

We adopt MSE to measure the distortion between the estimated function (i.e., $\hat{g}$) and the target function (i.e., $g$), which quantifies the performance of AirComp, given by
\begin{align}
\label{MSE_ori}
{\sf{MSE}}(\hat{g}, g) = \mathbb{E}\left(|\hat{g}-g|^2\right)
= \sum_{k\in\mathcal{K}}\left|   \frac{1}{\sqrt{\eta}}{\bm{m}}^{\sf{H}}(\bm  h_{\mathrm{d},k}+\bm G \bm\Theta \bm  h_{\mathrm{r},k} ){w}_k -1\right|^2 \!+\! \frac{\sigma^2\|\bm{m}\|^2}{\eta}. 
\end{align}
%The MSE given in \eqref{MSE_ori} depends on the transmit scalar $w_k$, the receive beamforming vector $\bm m$, the denoising factor $\eta$, and the phase-shift matrix $\bm \Theta$.
Motivated by \cite{chen2018uniform,multi_function, WPT}, we adopt the following zero-forcing design\footnote{It is clear that ${\sf{MSE}}(\hat{g}, g) \geq \frac{\sigma^2\|\bm{m}\|^2}{\eta}$.
Given $\bm m$ and $\bm\Theta$, the equality is achieved when $\{w_k\}_{k=1}^{K}$ are set according to \eqref{omega}, which enforces $\sum_{k\in\mathcal{K}}\left|   \frac{1}{\sqrt{\eta}}{\bm{m}}^{\sf{H}}(\bm  h_{\mathrm{d},k}+\bm G \bm\Theta \bm  h_{\mathrm{r},k} ){w}_k -1\right|^2$ to be zero. 
} to determine the transmit scalars of IoT devices  %$\forall ~ k \in \mathcal{K}$.
\begin{align}\label{omega}
w_k^{\star}=\sqrt{\eta}{{(\bm{m}^{\sf{H}}(\bm  h_{\mathrm{d},k}+\bm G \bm\Theta \bm  h_{\mathrm{r},k} ))^{\sf{H}}}\over{\|\bm{m}^{\sf{H}}(\bm  h_{\mathrm{d},k}+\bm G \bm\Theta \bm  h_{\mathrm{r},k} )\|^2}}, ~ \forall ~ k \in \mathcal{K}.
\end{align}
%\emph{Lemma 1 (Optimal Transmitter Scalar) :}
%Given an arbitrarily receive beamforming vector $\bm a$ and phase-shift matrix $\bm \Theta$ the computation error measured by MSE is minimized by the following zero-forcing transmitter:
%\begin{align}\label{omega}
%w_k^{\star}=\sqrt{\eta}{{(\bm{m}^{\sf{H}}(\bm  h_{\mathrm{d},k}+\bm G \bm\Theta \bm  h_{\mathrm{r},k} ))^{\sf{H}}}\over{\|\bm{m}^{\sf{H}}(\bm  h_{\mathrm{d},k}+\bm G \bm\Theta \bm  h_{\mathrm{r},k} )\|^2}},
%\end{align}
%where $\bm m$ denotes the normalized receiving beamforming vector, and $\eta$ is a power control factor that is chosen to satisfy the transmit power constraint, and the connection between $\bm m$ and $\bm a$ can be given as $\sqrt{\eta} \bm a = \bm m $.
%\begin{proof}
%Please refer to Appendix A for the proof.      
%%The proof is similar to the lemma 1 in \cite{WPT}.
%\end{proof}
Due to the maximum transmit power constraint of IoT devices, i.e., $|w_k|^2\leq P$, $\eta$ can be set as
\begin{align}\label{eta}
\eta=P\min_{k\in\mathcal{K}} \|\bm{m}^{\sf{H}}(\bm  h_{\mathrm{d},k}+\bm G \bm\Theta \bm  h_{\mathrm{r},k} )\|^2.
\end{align}
Therefore, the MSE at the AP can be further written as
\begin{align}
\label{MSE}
{\sf{MSE}}(\bm m, \bm \Theta)={{\|\bm{m}\|^2\sigma^2}\over{P\min_{k\in\mathcal{K}} \|\bm{m}^{\sf{H}}(\bm  h_{\mathrm{d},k}+\bm G \bm\Theta \bm  h_{\mathrm{r},k} )\|^2}}.
\end{align}

%To minimize the estimation error, we need to jointly optimize the transmit scalar $w_k$, the receive beamforming vector $\bm m$, power normalizing factor $\eta$, and the phase-shift matrix $\bm \Theta$.
%For simplicity, we adopt uniforming forcing design \cite{chen2018uniform} for the transmit scalar $w_k$, i.e.,

The MSE given in \eqref{MSE} is determined by the transmit signal-to-noise ratio (SNR) $\frac{P}{\sigma^2}$, the receive beamforming vector $\bm m$ at the AP, and the composite channel coefficients $\bm  h_{\mathrm{d},k}+\bm G \bm\Theta \bm  h_{\mathrm{r},k}$, $\forall ~ k$.

\begin{remark}
\emph{
        According to \eqref{MSE}, the MSE of AirComp is bottlenecked by the worst channel between the IoT devices and the AP.
        Without RIS, i.e., $\bf \Theta = 0$, the channel quality is only determined by the direct link.
        In this case, we can only adjust the transmit power of IoT devices to tackle the detrimental effects of severe channel fading and path loss.
%However, this may not always be practical as the maximum transmit power of IoT devices is limited.
With the assistance of RIS, the composite channel condition of each link (e.g., $\bm  h_{\mathrm{d},k}+\bm G \bm\Theta \bm  h_{\mathrm{r},k}$) can be adaptively adjusted by reconfiguring the phase-shift matrix $\bm\Theta$, which is able to enhance the channel quality of the link with the worse channel condition. 
As a result, the ability of RIS to reconfigure the propagation environment can be exploited to effectively mitigate the performance bottleneck of AirComp.}
\end{remark}

Our goal is to jointly optimize the phase-shift matrix $\bm \Theta$ and the receive beamforming vector $ \bm m $ to minimize ${\sf{MSE}}(\bm m, \bm \Theta)$ in \eqref{MSE}. 
%Hence, the formulated problem can be expressed as
It is formulated as the following problem
\begin{equation}
\label{ori_pro}
        \begin{split}
                \underset{\bm m,\bm \Theta}{\min}  ~ & \frac{\|\bm m\|^2 \sigma^2 }{P\min_{k\in\mathcal{K}}\|\bm m^{\sf H}(\bm  h_{\mathrm{d},k}+\bm G \bm\Theta \bm  h_{\mathrm{r},k} )\|^2} \\
                \text { s.t. } ~ & 0 \leq \theta_i < 2\pi, \forall ~  i. % =1,\ldots,N.
        \end{split}
\end{equation}
Problem \eqref{ori_pro} can be equivalently expressed as
\begin{equation}\label{eq:ori}
        \begin{split}
                \underset{\bm m,\bm \Theta}{\min} &~ \max_{k\in \mathcal{K}}
                \frac{\|\bm m\|^2 \sigma^2 }{P\|\bm m^{\sf H}(\bm  h_{\mathrm{d},k}+\bm G \bm\Theta \bm  h_{\mathrm{r},k} )\|^2} \\
                \text { s.t. } &~ 0 \leq \theta_i < 2\pi, \forall ~  i. %=1,\ldots,N.
        \end{split}
\end{equation}

Since $\bm \Theta$ is a diagonal matrix and $\bm  h_{\mathrm{r},k}$ is a vector, we rewrite $\bm  h_{\mathrm{d},k}+\bm G \bm\Theta \bm  h_{\mathrm{r},k} $ as $\bm  h_{\mathrm{d},k}+\bm G \text{diag}(\bm  h_{\mathrm{r},k})\bm v$, where $\bm v=[v_1, v_2, \ldots, v_N] $ with $v_i = e^{j\theta_i},~ \forall ~ i $. % $^{\sf T}=[e^{j\theta_1},e^{j\theta_2},\ldots,e^{j\theta_N}]^{\sf T}$.
As $\frac{\sigma^2}{P}$ is a constant, problem \eqref{eq:ori} is further equivalent to the following problem in the sense of having the same optimal solution
\begin{equation}
\label{eq:p0}
        \begin{split}
                \underset{\bm m,\bm v}{\min} &~ \max_k\frac{\|\bm m\|^2}{\|\bm m^{\sf H}(\bm  h_{\mathrm{d},k}+\bm G \text{diag}(\bm  h_{\mathrm{r},k})\bm v )\|^2}  \\
                \text { s.t. } &~ |v_i| = 1, \forall ~  i . %=1,\ldots,N.
        \end{split}
\end{equation}

%Problem \eqref{eq:p0} is a non-convex and highly \textcolor{red}{intractable to handle.}

\subsection{Limitations of State-of-the-Art Methods}
Most of the existing studies \cite{fu2020reconfigurable,jiang2019,zhibin} on the joint design in RIS-assisted wireless networks adopted the alternating SDR and alternating DC algorithms.
According to \cite{jiang2019}, problem \eqref{eq:p0} can be equivalently transformed to problem \eqref{key} as follows
\begin{equation}\label{key}
        \begin{split}
                \underset{\bm m,\bm v}{\min} ~ & \|\bm m \|^2\\
                \text { s.t. } ~ &|\bm m^{\sf H}(\bm G \text{diag}(\bm  h_{\mathrm{r},k})\bm v+\bm  h_{\mathrm{d},k}|^2\ge 1,\forall ~  k,\\
                \qquad & |v_i|=1,~\forall ~  i.
        \end{split}
\end{equation}
Problem \eqref{key} can be tackled by alternately solving the following two subproblems
\begin{equation}
        \begin{split}
\label{sdp1}
\underset{\bm m}{\min}&~\|\bm m \|^2\\
\text { s.t. }&~|\bm m^{\sf H}\bm h_k^e|^2\ge 1,\forall ~  k,
\end{split}
\end{equation}
and
\begin{equation}
\begin{split}
\label{sdp2}
\operatorname{find}&~\bm v\\
\text { s.t. }&~|\bm a^{\sf H}_k\bm v+c_k |^2\ge 1,\forall ~  k,\\
&~|v_i|^2=1,\forall ~  i,
\end{split}
\end{equation}
%\begin{subequations}
%\begin{align}
%\begin{split}
%\label{sdp1}
%\underset{\bm m}{\min}&~\|\bm m \|^2\\
%\text { s.t. }&~|\bm m^{\sf H}\bm h_k^e|^2\ge 1,\forall ~  k,
%\end{split} \\
%\begin{split}
%\label{sdp2}
%\text{and}
%\operatorname{find}&~\bm v\\
%\text { s.t. }&~|\bm a^{\sf H}_k\bm v+c_k |^2\ge 1,\forall ~  k,\\
%&~|v_i|^2=1,\forall ~  i,
%\end{split}
%\end{align}    
%\end{subequations}
where $\bm h_k^e = \bm G \text{diag}(\bm  h_{\mathrm{r},k})\bm v+\bm h_{\mathrm{d},k}$, $ \bm a^{\sf H}_k =  \bm m^{\sf H}\bm G\text{diag}(\bm  h_{\mathrm{r},k}) $, and $c_k = \bm m^{\sf H}\bm h_{\mathrm{d},k}$.
The two non-convex quadratically constrained quadratic programming (QCQP) problems, i.e., \eqref{sdp1} and \eqref{sdp2}, were then converted into two SDP problems with rank-one constraints by using the matrix lifting approach.
An intuitive solution is to apply the SDR technique \cite{Tom_luo_sdr} to convexify the problems by directly dropping the rank-one constraints, yielding the alternating SDR algorithm \cite{zhibin}. 
As an alternative, based on the fact that the rank-one constraint of a positive definite matrix is equivalent to the zero-difference between the spectral norm and trace norm,
a DC technique \cite{DC} was proposed to tackle the rank-one constraint, yielding the alternating DC algorithm \cite{jiang2019}.
The existing studies that adopted the aforementioned framework suffer from two limitations, i.e., non-guaranteed convergence and high computation complexity.
Specifically,
the optimization of phase-shift vector $\bm v$ involves a feasibility detection problem, i.e., \eqref{sdp2}.
According to the analysis in \cite{Yu2020}, both the alternating SDR and alternating DC algorithms are not guaranteed to converge.
Besides, AirComp is expected to support wireless data aggregation in high-density IoT networks, where the number of IoT devices would be large.
However, the aforementioned state-of-the-art methods, i.e., alternating DC and alternating SDR  algorithms, are not scalable because of their high computational complexity.
Especially for the alternating DC algorithm, a series of SDP problems generated by the SCA technique need to be solved by the standard interior-point method \cite{Nesterov_interior} at each alternating iteration.
The computation time consumption would be an unaffordable burden when the aforementioned algorithms are applied to solve large-scale optimization problems. 
The limitations of the existing studies motivate us to develop a low-complexity algorithm with a theoretical convergence guarantee for RIS-assisted AirComp systems.

\subsection{Problem Transformation}

To mitigate the aforementioned limitations, in this subsection, we reformulate problem \eqref{eq:p0} as a min-max optimization problem, which is presented in the following proposition.

\begin{proposition}
\label{equi}
\emph{
Problem \eqref{eq:p0} is equivalent to the following min-max QCQP problem:
\begin{equation}\label{eq:p1}
        \begin{split}
                \underset{\bm m,\bm v}{\min} &~ \max_k \Big\{-\|\bm m^{\sf H}(\bm  h_{\mathrm{d},k}+\bm G \mathrm{diag}(\bm  h_{\mathrm{r},k})\bm v )\|^2\Big\} \\
                \text { s.t. } &~ |v_i| = 1,\forall ~  i,\\ %=1,\ldots,N,\\
                 &~ \|\bm m\|^2 = 1.
        \end{split}
\end{equation}
}
\end{proposition}
\begin{proof}
  Please refer to Appendix A.   
\end{proof}
As can be observed from \eqref{eq:p1}, both optimization variables $\bm m$ and $\bm v$ are involved in the objective function.
This transformation enables us to eliminate the feasibility detection problem as in \cite{jiang2019}, and allows us to exploit the monotonicity of the objective function during alternating minimization.
%The transformation enables us to eliminate the feasibility detection problem compared to \cite{jiang2019} while we adopt the alternate optimization method for solving $\bm m$ and $\bm v$, which will be revealed in the next section.
%Besides, we can build the monotonicity of the objective function when the later proposed algorithm be applied under this framework.
Problem \eqref{eq:p1} is still a challenging non-convex optimization problem. 
Specifically, solving problem \eqref{eq:p1} faces the following three challenges.
First, the optimization variables $\bm v$ and $\bm m$ are coupled in the objective function.
Second, the unit-modulus constraint on $v_i,~ \forall ~ i$ and the unit receive power constraint on $\bm m$ are non-convex.
Third,  the pointwise maximum of quadratic terms, i.e., the objective function,  is non-convex and non-smooth. 
To tackle these issues, we shall propose an alternating minimization method in conjunction with SCA to solve problem \eqref{eq:p1} in the following section.

\section{Proposed AlterMin SCA Framework}
In this section, we propose an alternating minimization method to alternately optimize the receive beamforming vector $\bm m$ and the phase-shift vector $\bm v$, resulting in two non-convex subproblems with respect to $\bm m$ and $\bm v$, respectively.
We then construct convex approximations for the two yielded subproblems by using SCA.
%And we develop a SCA algorithm to solve the two generated subproblems.
%In this section, we propose to alternately optimize $\bm m$ and $\bm v$ with the other fixed. 
%As discussed in last section, they can be cast to the identical form \eqref{general_pro}.
%To solve the two problem \eqref{m:real:quad} and problem \eqref{v:real:quad}, we will first apply successive convex approximation technique to approximate the nonconvex objective function.
%Specifically, we iteratively replace the concave function by its linear approximation, which 
%       is also an upper bound of the concave function.
%The piece-maximum of linear functions is a non-smooth but convex function.
%It means that the subproblem is a non-smooth convex problem.
%Though which can be solved directly by some convex programming tools, like CVX.
%However, the efficiency cannot be guarantee \cite{Fast}.
%And there existing nonsmooth convex optimization algorithm like subgradient descent, mirror descent method can work on it, but the convergence rate is only $\frac{1}{\sqrt{t}}$ \cite{firstorder}.
%We then propose to apply Mirror-Prox algorithm to solve problem \eqref{m:real:quad} and problem \eqref{v:real:quad} which has better convergence performance compare to the aforementioned algorithms.
%Next, we introduce the implementation details of the Mirror-Prox algorithm.
%In the final of this section, we summarize the proposed Alternating Mirror-Prox Based SCA Algorithm and analyse its computation complexity. 

\subsection{Phase-Shift Vector Optimization}
\label{phase-shift-opt}
When the receive beamforming vector $\bm m$ is fixed, problem \eqref{eq:p1} is reduced to the following subproblem that requires the optimization of phase-shift vector $\bm v$
\begin{equation}\label{sub2}
        \begin{split}
                \underset{\bm v}{\min} &~ \max_k ~ \left\{-\|\bm m^{\sf H}(\bm  h_{\mathrm{d},k}+\bm G \text{diag}(\bm  h_{\mathrm{r},k})\bm v )\|^2 \right\} \\
                \text { s.t. } &~ |v_i| = 1,\forall ~  i=1,\ldots,N.
       \end{split}
\end{equation}
By denoting $c_k = \bm m^{\sf H}\bm  h_{\mathrm{d},k}$ and $\bm a_k^{\sf H} = \bm m^{\sf H}\bm G \text{diag}(\bm  h_{\mathrm{r},k})$,
problem \eqref{sub2} can be rewritten as 
\begin{equation}\label{sub2:neg}
        \begin{split}
             \underset{\bm v}{\min} &~ \max_k \left\{-\|c_k + \bm a_k^{\sf H}\bm v \|^2 \right\} \\
             \text { s.t. } &~ |v_i| = 1,\forall ~  i.
        \end{split}
\end{equation}
%By convention, we are custom to deal with the min-max form problem.
%Therefore, we further transform problem \eqref{re:sub2} to its equivalent problem in the min-max form as follows 
%\[
%\begin{aligned}
%        \begin{split}
%                \underset{\bm {v}}{\max} &~ \min_k \|c_k + \bm a_k^{\sf H}\bm v \|^2 \\
%                \text { s.t. } &~ |v_i| = 1,\forall ~  i.
%       \end{split} & \Longleftrightarrow &
%         \begin{split}
%                \underset{\bm {v}}{\min} &~ \Big\{-\min_k \|c_k + \bm a_k^{\sf H}\bm v \|^2 \Big\} \\
%                \text { s.t. } &~ |v_i| = 1,\forall ~  i.
%       \end{split} \\
%       & \Longleftrightarrow &
%        \begin{split}
%             \underset{\bm v}{\min} &~ \max_k \left\{-\|c_k + \bm a_k^{\sf H}\bm v \|^2 \right\} \\
%             \text { s.t. } &~ |v_i| = 1,\forall ~  i.
%        \end{split}
%\end{aligned}
%\]
For simplicity of algorithm design, we further convert problem \eqref{sub2:neg} from the complex domain to the real domain.
By denoting $\tilde{\bm v} = [\Re\{\bm v\}^{\sf T},\Im \{\bm v\}^{\sf T}]^{\sf T} \in \mathbb{R}^{2N}$, we obtain the following problem 
\begin{equation}\label{v:real0:quad}
        \begin{split}
                \underset{\tilde{\bm v}}{\min} &~ \max_k \Big\{ \tilde{\bm v}^{\sf T}\tilde{\mathbf{A}}_{k}\tilde{\bm v} -
                 2\tilde{\bm v}^{\sf T} \bm b_k - |c_k|^2\Big\} \\ % -|c_k|^2 
                \text { s.t. } &~ {\tilde{v}}_i^2+{\tilde{v}}_{i+N}^2 = 1,\forall ~ i,
       \end{split}
\end{equation}
where $ \bm b_k = [\Re\{c_k\bm a_k\}^{\sf T},\Im\{c_k\bm a_k\}^{\sf T}]^{\sf T}$ and
\begin{equation}
\begin{aligned}
% \bm b_k =& [\Re\{c_k\bm a_k\}^{\sf T},\Im\{c_k\bm a_k\}^{\sf T}]^{\sf T}, \\
\tilde{\mathbf{A}}_{k} =
        \begin{bmatrix}
                \Re\left\{-\bm a_k \bm a_k^{\sf H}\right\}& \qquad  &-\Im \left\{-\bm a_k \bm a_k^{\sf H}\right\} \\
        \Im \left\{-\bm a_k \bm a_k^{\sf H}\right\}& \qquad  & \Re\left\{-\bm a_k \bm a_k^{\sf H}\right\}  
        \end{bmatrix}.
\end{aligned} \nonumber
\end{equation}
Problem \eqref{v:real0:quad} aims to minimize the pointwise maximum of concave quadratic terms.
Solving problem \eqref{v:real0:quad} is challenging due to the non-convex constraints and the non-convex objective function.
In the following, we tackle the non-convex constraint by utilizing the convex relaxation technique.
Specifically, we relax the unit modulus constraint to $\tilde{\bm v} \in \mathcal{V}$,
where $\mathcal{V} = \{\tilde{\bm v} \mid {\tilde{v}}_i^2+{\tilde{v}}_{i+N}^2 \leq 1,\forall ~  i=1,\ldots,N\}$, yielding the following relaxed problem
\begin{equation}\label{v:real:quad}
        \begin{split}
                \underset{\tilde{\bm v}}{\min} &~ \max_k \left\{ \tilde{\bm v}^{\sf T}\tilde{\mathbf{A}}_{k}\tilde{\bm v} -
                 2\tilde{\bm v}^{\sf T} \bm b_k -|c_k|^2\right\} \\ % -|c_k|^2 
                \text { s.t. } &~ \tilde{\bm v} \in \mathcal{V}.
       \end{split}
\end{equation}
On the other hand, the SCA technique is applied to tackle the non-convexity of the objective function 
$\max_k \Big\{ \tilde{\bm v}^{\sf T}\tilde{\mathbf{A}}_{k}\tilde{\bm v} -
                 2\tilde{\bm v}^{\sf T} \bm b_k -|c_k|^2\Big\}$.
                 In particular, due
%\subsubsection{SCA for Phase-Shift Optimization}
%\label{sca_text}
%As the objective function is non-convex,
%we adopt an iterative procedure to approximate the original non-convex
%objective function by convex function, and then solve a sequence of convex problems, such method is SCA.
%SCA is widely used technique for building local convexity for non-convex optimization problems.
%2-26
%The principle of SCA method is to replace the original non-convex function by properly chosen convex surrogate based on the current iterate, and then solve the formulated subproblem.
%The major procedure of SCA for problem \eqref{v:real:quad} is presented as follows.
%Due 
to the concavity of $\{\tilde{\bm v}^{\sf T}\tilde{\mathbf{A}}_{k}\tilde{\bm v} -
                 2\tilde{\bm v}^{\sf T} \bm b_k -|c_k|^2\}$, we construct its linear upper bound based on the first-order Taylor approximation as
\[
\tilde{\bm v}^{\sf T}\tilde{\mathbf{A}}_{k}\tilde{\bm v} -
                 2\tilde{\bm v}^{\sf T} \bm b_k -|c_k|^2 
                  \leq 
                 \left(\mathbf{p}_{k}^{(n)}\right)^{\sf T} \tilde{\bm v} + q_{k}^{(n)},
\]
where $\mathbf{p}_{k}^{(n)} = 2(\tilde{\mathbf{A}}_{k}\tilde{\bm v}^{(n)} -
                 \bm b_k), ~
        q_{k}^{(n)} = - (\tilde{\bm v}^{(n)})^{\sf T}\tilde{\mathbf{A}}_{k}\tilde{\bm v}^{(n)}-|c_k|^2,$ and $\tilde{\bm v}^{(n)}$ is the solution obtained at the $n$-th iteration.
As a result, we have
\begin{equation}
\label{first}
\max_k \left\{ \tilde{\bm v}^{\sf T}\tilde{\mathbf{A}}_{k}\tilde{\bm v} -
                 2\tilde{\bm v}^{\sf T} \bm b_k -|c_k|^2 \right\} 
                 \leq
                 \max_k \left\{  \left(\mathbf{p}_{k}^{(n)}\right)^{\sf T} \tilde{\bm v} + q_{k}^{(n)}\right\}.
\end{equation}
Therefore, at the $(n+1)$-th iteration, we can replace the non-convex objective function in problem \eqref{v:real:quad} by its convex surrogate $\max_k ~ \left\{  \left(\mathbf{p}_{k}^{(n)}\right)^{\sf T} \tilde{\bm v} + q_{k}^{(n)}\right\}$.
Specifically, at the $(n\!+\!1)$-th iteration, problem \eqref{v:real:quad} is approximated by the following subproblem 
\begin{equation}\label{saddle:ori2}
        \begin{split}
                \underset{\tilde{\bm v}}{\min} ~& \max_k ~  \left(\mathbf{p}_{k}^{(n)}\right)^{\sf T} \tilde{\bm v} + q_{k}^{(n)} \\
                \text { s.t. } ~& \tilde{\bm v} \in \mathcal{V}.
       \end{split}
\end{equation}
To efficiently solve problem \eqref{v:real:quad}, we resort to iteratively solve its non-smooth convex approximation problem \eqref{saddle:ori2}.
%The overall SCA algorithm for solving problem \eqref{v:real:quad} is summarized in Algorithm \ref{SCA_V_a}.
%And the convergence property of the SCA algorithm for such problem have been drawn in \cite{Fast}. 
%\begin{algorithm} \label{SCA_V_a}
%\caption{SCA method for problem \eqref{v:real:quad}}
%\SetKwRepeat{Repeat}{repeat}{until}
%%\SetKwFor{For}{For}{end}
%\SetAlgoLined
%\KwIn{Initia point $\tilde{\bm v}^{(0)} \in \mathcal{V}$, and threshold $\epsilon$, and set $n:=0$;}
%%\KwOut{BB}
%\Repeat{Decrease of the objective value is less than threshold $\epsilon$}
%{
%       $
%       \tilde{\bm v}^{(n+1)} = \argmin\limits_{\tilde{\bm v} \in \mathcal{V}} \! ~ \! \max\limits_k ~  \left(\mathbf{p}_{k}^{(n)}\right)^{\sf T} \tilde{\bm v} + q_{k}^{(n)}
%       $; \\
%       Update $\mathbf{p}_{k}^{(n+1)}$ and $q_{k}^{(n+1)}$; \\ %  by \eqref{p_q}
%       $n \leftarrow n+1$;
%}
%\KwOut{Take $\tilde{\bm v}^{(n+1)}$ as an approximation solution.}
%\end{algorithm}

\subsection{Receive Beamforming Vector Optimization}

When the phase-shift vector $\bm v$ is fixed, 
\eqref{eq:p1} can be formulated as an optimization problem with respect to the receive beamforming vector $\bm m$ as follows
\begin{equation}\label{sub1}
        \begin{split}
                \underset{\bm m}{\min} &~ \max_k \left\{-\|\bm m^{\sf H}(\bm  h_{\mathrm{d},k}+\bm G \text{diag}(\bm  h_{\mathrm{r},k})\bm v)\|^2 \right\} \\
                \text { s.t. } &~ \|\bm m\|^2 = 1.
        \end{split}
\end{equation}
By denoting $\bm{h}_{k} = \bm  h_{\mathrm{d},k}+\bm G \text{diag}(\bm  h_{\mathrm{r},k})\bm v$, \eqref{sub1} can be represented as 
\begin{equation}\label{re:sub1}
        \begin{split}
             \underset{\bm m}{\min} &~ \max_k ~ \left\{-\|\bm m^{\sf H}\bm{h}_{k}\|^2 \right\} \\
                \text{s.t.} &~ \|\bm m\|^2 = 1.
        \end{split}
\end{equation}
The constraint in \eqref{re:sub1} is non-convex.
According to \cite{equivalent}, 
%there exists a convex constraint that has an equivalent effect as the non-convex constraint in problem \eqref{re:sub1}.
%Specifically,
problem \eqref{re:sub1} is equivalent to the following problem 
\begin{equation}
\label{re:sub1:ineq}
        \begin{split}
             \underset{\bm m}{\min} &~ \max_k ~ \left\{-\|\bm m^{\sf H}\bm{h}_{k}\|^2 \right\} \\
                \text{s.t.}~& \|\bm m\|^2 \leq 1.
        \end{split}
\end{equation}
This is because that the constraint should be met with equality at the
optimal point for problem \eqref{re:sub1:ineq}.
Otherwise, $\bm m$ could be scaled up to reduce the objective value, thereby contradicting the optimality.
%Similar to the transformation of max-min to min-max for the phase-shift vector optimization, we have 
%\begin{equation}
%\label{transform_for_m}
%\begin{aligned}
%        \begin{split}
%             \underset{\bm m}{\max} ~& \min_k  \|\bm m^{\sf H}\bm{h}_{k}\|^2\\
%                \text { s.t. } ~& \|\bm m\|^2 \leq 1.
%        \end{split} \Longleftrightarrow 
%        \begin{split}
%             \underset{\bm m}{\min} ~& \max_k  -\|\bm m^{\sf H}\bm{h}_{k}\|^2\\
%                \text { s.t. } ~&\|\bm m\|^2 \leq 1.
%        \end{split} 
%\end{aligned}
%\end{equation}
By defining $\tilde{\bm m} = [\Re\{\bm m\}^{\sf T},\Im \{\bm m\}^{\sf T}]^{\sf T} \in \mathbb{R}^{2M}$, we convert problem \eqref{re:sub1:ineq} from the complex domain to the real domain to facilitate the algorithm design  
\begin{equation}\label{m:real:quad}
        \begin{split}
             \underset{\tilde{\bm m}}{\min} ~ & \max_k  ~ \tilde{\bm m}^{\sf T}\tilde{\mathbf{H}}_{k} \tilde{\bm m}\\
                \text { s.t. } ~& \tilde{\bm m} \in \mathcal{M},
        \end{split}
\end{equation}
\noindent where $\mathcal{M} = \{\tilde{\bm m} \mid \|\tilde{\bm m}\|^2 \leq 1\}$ and
\begin{equation}
\tilde{\mathbf{H}}_{k}=\begin{bmatrix}
\Re\left\{-\bm{h}_{k}\bm{h}_{k}^{\sf H}\right\} & \qquad  & -\Im \left\{-\bm{h}_{k}\bm{h}_{k}^{\sf H}\right\} \\
\Im \left\{-\bm{h}_{k}\bm{h}_{k}^{\sf H}\right\} & \qquad  & \Re\left\{-\bm{h}_{k}\bm{h}_{k}^{\sf H}\right\}
\end{bmatrix}. \nonumber
\end{equation} 
%As a result, problem \eqref{re:sub1:ineq} can be rewritten as 
%\begin{equation}\label{ori:saddle1}
%        \begin{split}
%             \underset{\bm m}{\min} ~& \max_k ~ \bm m^{\sf H} \mathbf{H}_{k} \bm m\\
%                \text { s.t. } ~& \|\bm m\|^2 \leq 1,
%        \end{split}
%\end{equation}
%\noindent where $\mathbf{H}_{k} = -\bm{h}_{k}\bm{h}_{k}^{\sf H}$.
%By defining $\tilde{\bm m} = [\Re\{\bm m\}^{\sf T},\Im \{\bm m\}^{\sf T}]^{\sf T} \in \mathbb{R}^{2M}$,
%problem \eqref{ori:saddle1} can be equivalently transformed to 
%\begin{equation}\label{m:real:quad}
%        \begin{split}
%             \underset{\tilde{\bm m}}{\min} ~ & \max_k  ~ \tilde{\bm m}^{\sf T}\tilde{\mathbf{H}}_{k} \tilde{\bm m}\\
%                \text { s.t. } ~& \tilde{\bm m} \in \mathcal{M},
%        \end{split}
%\end{equation}
%where $\mathcal{M} = \{\tilde{\bm m} \mid \|\tilde{\bm m}\|^2 \leq 1\}$ and
%\begin{equation}
%\tilde{\mathbf{H}}_{k}=\begin{bmatrix}
%\Re\left\{\mathbf{H}_{k}\right\} & \qquad  & -\Im \left\{\mathbf{H}_{k}\right\} \\
%\Im \left\{\mathbf{H}_{k}\right\} & \qquad  & \Re\left\{\mathbf{H}_{k}\right\}
%\end{bmatrix}. \nonumber
%\end{equation} 
To tackle the non-convexity of the objective function of problem \eqref{m:real:quad}, we shall apply SCA to construct its local convex approximation. 
Due to the similar structure of problem \eqref{v:real:quad} and problem \eqref{m:real:quad},
the derivation here is similar to that presented in Section \ref{phase-shift-opt}.
For completeness, we sketch the main procedures for solving problem \eqref{m:real:quad}.  
Starting from an initial point $\tilde{\bm m}^{(0)} \in \mathcal{M}$, SCA is applied to generate a sequence of solutions $\{\tilde{\bm m}^{(n)}\}$ as follows.
With the approximated solution $\tilde{\bm m}^{(n)}$ obtained at the $n$-th iteration, the concave quadratic function $\tilde{\bm m}^{\sf T}\tilde{\mathbf{H}}_{k} \tilde{\bm m}$ can be upper bounded by its linear majorization.
Specifically, we have the following inequality 
\[
\tilde{\bm m}^{\sf T}\tilde{\mathbf{H}}_{k} \tilde{\bm m} \leq (2\tilde{\mathbf{H}}_{k} \tilde{\bm m}^{(n)})^{\sf T}\tilde{\bm m} - (\tilde{\bm m}^{(n)})^{\sf T}\tilde{\mathbf{H}}_{k} \tilde{\bm m}^{(n)}.
\]
By denoting $\mathbf{\bar{p}}_{k}^{(n)} = 2\tilde{\mathbf{H}}_{k} \tilde{\bm m}^{(n)}$ and $\bar{q}_{k}^{(n)} = - (\tilde{\bm m}^{(n)})^{\sf T}\tilde{\mathbf{H}}_{k} \tilde{\bm m}^{(n)}$, we have
\begin{equation}
\max_k  ~ \left\{ \tilde{\bm m}^{\sf T}\tilde{\mathbf{H}}_{k} \tilde{\bm m} \right\}
\leq 
\max_k ~ \left\{ \mathbf{\bar{p}}_{k}^{(n) \sf T} \tilde{\bm m} + \bar{q}_{k}^{(n)}\right\}.
\end{equation}
%Note that $\max_k \left\{\mathbf{\bar{p}}_{k}^{(n) \sf T} \tilde{\bm m} \!+\! \bar{q}_{k}^{(n)}\right\}$ is a convex upper bound of $\max_k \left\{\tilde{\bm m}^{\sf T}\tilde{\mathbf{H}}_{k} \tilde{\bm m}\right\}$ based on $\tilde{\bm m}^{(n)}$.
Then, $\tilde{\bm m}^{(n+1)}$ can be obtained by solving the following non-smooth convex approximation problem
\begin{equation}
\label{SCA_m}
        \begin{split}
                \underset{\tilde{\bm m}}{\min} ~& \max_k ~ \mathbf{\bar{p}}_{k}^{(n) \sf T} \tilde{\bm m} \!+\! \bar{q}_{k}^{(n)} \\
                \text { s.t. }~&\tilde{\bm m} \in \mathcal{M}.
       \end{split}
\end{equation}
%
%\begin{equation}
%\label{SCA_m}
%        \begin{split}
%               \tilde{\bm m}^{(n+1)} =  \underset{\tilde{\bm m}}{\min} ~& \max_k ~ \mathbf{\bar{p}}_{k}^{(n) T} \tilde{\bm m} + \bar{q}_{k}^{(n)} \\
%                \text { s.t. } ~& \tilde{\bm m} \in \mathcal{M},
%       \end{split}
%\end{equation}

%The overall algorithm is summarized in Algorithm \eqref{SCA_Algo_M}.
%
%\begin{algorithm} \label{SCA_Algo_M}
%\caption{SCA method for problem \eqref{m:real:quad}}
%\SetKwRepeat{Repeat}{repeat}{until}
%%\SetKwFor{For}{For}{end}
%\SetAlgoLined
%\KwIn{Initialized point $\tilde{\bm m}^{(0)} \in \mathcal{M}$}
%%\KwOut{BB}
%\For{n=0,1,2,\ldots}{
%       $
%       \tilde{\bm m}^{(n+1)} = \argmin_{\tilde{\bm m} \in \mathcal{M}} \max_k \mathbf{\bar{p}}_{k}^{(n) T} \tilde{\bm m} \!+\! \bar{q}_{k}^{(n)}
%       $ \\
%       Update $\mathbf{\bar{p}}_{k}^{(n)}$ and $\bar{q}_{k}^{(n+1)}$ %  by \eqref{p_q}
%}
%\end{algorithm}
%
%\subsection{Overall Algorithm}

\subsection{Convergence Analysis}

We recall that problem \eqref{eq:p1} is decomposed into two problems \eqref{v:real:quad} and \eqref{m:real:quad} with respect $\tilde{\bm v}$ and $\tilde{\bm m}$, respectively, which are then alternately solved by using SCA.
Finally, we project $\tilde{\bm v}$ to $\mathcal{V}_1=\{\tilde{\bm v} \mid {\tilde{v}}_i^2+{\tilde{v}}_{i+N}^2 = 1,\forall ~  i\}$, so as to compensate the relaxation on $\tilde{\bm v}$.  
The overall AlterMin SCA algorithm for solving problem \eqref{eq:p1} is summarized in Algorithm 1. 
%There are two main procedures that is alternatively optimizing $\tilde{\bm m}$ and $\tilde{\bm v}$ by using SCA method.   
The convergence of Algorithm 1 is presented in the following proposition.
\begin{proposition}
\label{con}
\emph{
The convergence property of Algorithm 1 consists of two parts:\\
i) In the inner loop (Steps 4-7 and Steps 10-13), i.e., SCA iteration, the objective values of problems \eqref{v:real:quad} and \eqref{m:real:quad} achieved by sequences $\{\tilde{\bm v}^{(n)}_{(l)}\}_{n=0}^{\infty}$ and $\{\tilde{\bm m}^{(n)}_{(l)}\}_{n=0}^{\infty}$ establish non-increasing convergent sequences;\\
ii) In the outer loop (Steps 2-16), i.e., alternating minimization iteration, the objective value of problem \eqref{eq:p1} achieved by the sequence $\{\tilde{\bm v}^{(0)}_{(l)},\tilde{\bm m}^{(0)}_{(l)}\}_{l=0}^{\infty}$ establish a non-increasing convergent sequence.
}
\end{proposition}
%The proof can be found in Appendix B.
\begin{proof}
Please refer to Appendix B.
\end{proof}

%Problem \eqref{eq:p1} is decomposed into two subproblems with respect to $\bm v$ and $\bm m$, i.e., \eqref{re:sub2} and \eqref{re:sub1}, respectively. 
%To obtain a solution for problem \eqref{eq:p1}, we alternately solve problem \eqref{re:sub2} and problem \eqref{re:sub1}.
%Both of the two problems involve the max-min optimization problem. 

\begin{algorithm}[t] 
\label{altermin}
        \caption{AlterMin SCA for problem \eqref{eq:p1}}
        \SetKwRepeat{Repeat}{repeat}{until}
%\SetKwFor{For}{For}{end}
\SetAlgoLined
\KwIn{Initial point $\tilde{\bm v}^{(0)}_{(0)}$, $\tilde{\bm m}^{(0)}_{(0)}$ and threshold $\epsilon$;}
%\KwOut{BB}
Set: $l=0$;\\
\Repeat{the decrease of objective value $\Big\{-\|\bm m^{\sf H}(\bm  h_{\mathrm{d},k}+\bm G \mathrm{diag}(\bm  h_{\mathrm{r},k})\bm v )\|^2\Big\}$ less than $\epsilon$}{
Set: $n=0$;\\
\Repeat{the decrease of objective value of $\tilde{\bm v}^{\sf T}\tilde{\mathbf{A}}_{k}\tilde{\bm v} -
                 2\tilde{\bm v}^{\sf T} \bm b_k -|c_k|^2$ less than $\epsilon$}{
%$\bm z_0 = [\left(\tilde{\bm v}^{(0)}\right)^{\mathsf T},\left(\tilde{\bm y}^{(0)}\right)^{\mathsf T}]^{\mathsf T}$
Update $\tilde{\bm v}^{(n+1)}_{(l)}$ by solving problem \eqref{saddle:ori2};\\
$n \leftarrow n+1$;
}
Set: $\tilde{\bm v}^{(0)}_{(l+1)} = \tilde{\bm v}_{(l)}^{(n)}$;\\
Set: $n=0$;\\
\Repeat{the decrease of objective value of $\tilde{\bm m}^{\sf T}\tilde{\mathbf{H}}_{k} \tilde{\bm m}$ less than $\epsilon$}{
%$\bm z_0 = [\left(\tilde{\bm v}^{(0)}\right)^{\mathsf T},\left(\tilde{\bm y}^{(0)}\right)^{\mathsf T}]^{\mathsf T}$
Update $\tilde{\bm m}^{(n+1)}_{(l)}$ by solving problem \eqref{SCA_m};
\\
$n \leftarrow n+1$;
}
Set: $\tilde{\bm m}^{(0)}_{(l+1)} = \tilde{\bm m}^{(n)}_{(l)}$;\\
$l \leftarrow l+1$;
%\\${MSE}_{(l)}=\min_k \|{\bm m^{(n)}_{(l)}}^{\sf H}(\bm  h_{\mathrm{d},k}+\bm G \mathrm{diag}(\bm  h_{\mathrm{r},k}){\bm v^{(n)}_{(l)}})\|^2$
}
Project $\tilde{\bm v}$ to set $\mathcal{V}_1$.
\end{algorithm}

%\begin{algorithm}
%\SetAlgoLined
%\SetKwInput{KwInput}{Input}                % Set the Input
%\SetKwInput{KwOutput}{Output}              % set the Output
%\KwInput{initial point $\bm m_0$ and set $ j := 0$.}
% \For{j=1,2,3,\ldots}{
% update $\bm v^{j+1}$ by solving problem \eqref{v:real:quad}.\\
% update $\bm m^{j+1}$ by solving problem \eqref{m:real:quad}.
%    }
% \label{Alter}
% \caption{Alternating Minimization}
%\end{algorithm}
\subsection{Algorithm Discussion}
\label{discussion}
To efficiently solve problem \eqref{eq:p1}, we still need to design an efficient algorithm to solve problem \eqref{saddle:ori2} and problem \eqref{SCA_m}.
However, the objective function $\max_k \{ (\mathbf{p}_{k}^{(n)})^{\sf T} \tilde{\bm v} + q_{k}^{(n)}\}$ of problem \eqref{saddle:ori2} is convex but non-smooth.
Intuitively, the subgradient algorithm as a first-order method can be applied. However, it requires $\mathcal{O}\left(\frac{1}{\epsilon^2} \right)$ iterations to attain an $\epsilon$-optimal solution \cite{Bubeck2015}. 
Besides, the Nesterov's smoothing technique in conjunction with the accelerated gradient algorithm can also be employed to solve problem \eqref{saddle:ori2}.  It attains an $\epsilon$-optimal solution with $\mathcal{O}\left( \sqrt{\frac{1}{\epsilon \mu_s}} \right)$ iterations \cite{nesterov2005smooth}, where $\mu_s$ is the smoothness parameter.
Notice that the convergence performance of Nesterov's approach is very sensitive to the smoothness parameter (i.e., $\mu_s$), whose optimal value is, in general, difficult to determine.
In addition, by introducing an auxiliary variable $s$, problem \eqref{saddle:ori2} can be equivalently formulated as the following convex QCQP problem
\begin{equation}\label{socp}
        \begin{split}
                \underset{\tilde{\bm v}, s}{\min}  ~~& s \\
                \text { s.t. }  ~&   \left(\mathbf{p}_{k}^{(n)}\right)^{\sf T} \tilde{\bm v} + q_{k}^{(n)} - s  \leq 0, \forall ~  k, \\
                                ~& \tilde{\bm v} \in \mathcal{V}, ~ s \in \mathbb{R}.
       \end{split}
\end{equation}
Problem \eqref{socp} can be solved by using the interior-point method \cite{Nesterov_interior}, which attains an $\epsilon$-optimal solution with only $\mathcal{O}\left(\sqrt{N+K}\log \frac{2(N+K)}{\epsilon}\right)$ iterations. However, the time complexity of each iteration is $\mathcal{O}\left((N+K)N^{2} + N^3 \right)$ \cite{Nesterov_interior}.
As problem \eqref{saddle:ori2} and problem \eqref{SCA_m} have a similar form, the above analysis also applies to problem \eqref{SCA_m}.
Hence, all of the aforementioned algorithms cannot efficiently solve our problem when the number of optimization variables is large.
%A highly efficient algorithm for solving problem \eqref{saddle:ori2} should own both the high convergence rate and low complexity at each iteration.
This motivates us to exploit the underlying structure of problems \eqref{saddle:ori2} and \eqref{SCA_m} to develop a highly efficient algorithm with a fast convergence rate and low iteration cost in the following section.

\section{Mirror-Prox for Non-Smooth Convex Problems}

In this section, we aim to develop a low-complexity algorithm to solve the non-smooth convex problems \eqref{saddle:ori2} and \eqref{SCA_m}.
We equivalently convert problems \eqref{saddle:ori2} and \eqref{SCA_m} to the smooth convex-concave saddle point problems by using the primal-dual transformation,
and then propose to use the Mirror-Prox method \cite{variational} to solve the resulting problems.
%The convergence rate $\mathcal{O}\left(\frac{1}{\epsilon} \right)$ of Mirror-Prox algorithm is established in \cite{variational}.
%Besides, due to the special form of our problem, the low complexity of each iteration can be guaranteed.
%The details are presented as follows.

\subsection{Mirror-Prox Method for Non-smooth Convex Problem \eqref{saddle:ori2}}
%The Mirror-Prox method was firstly proposed in \cite{variational} for solving the Lipschitz continuous variational inequality problem.
%In this subsection, we show that problem \eqref{saddle:ori2} can be cast as a saddle point problem that corresponds to a variational inequality problem, 
%and then apply the Mirror-Prox method to solve the problem.
 
\subsubsection{Smooth Saddle Point Problem Formulation}

The objective function of problem \eqref{saddle:ori2} is pointwise maximum of affine functions.
We equivalently convert non-smooth problem \eqref{saddle:ori2} to a smooth convex-concave saddle point problem in Lemma \ref{P_D}. 

\begin{lemma}
\label{P_D}
\emph{
 (Primal-Dual Transformation) The non-smooth convex problem \eqref{saddle:ori2} is equivalent to the following smooth convex-concave saddle point problem       
 \begin{equation}\label{ori:saddle4}
        \begin{split}
             \underset{\tilde{\bm v}}{\min} &~ \max_{\mathbf{y}} ~ \left(\mathbf{P}^{(n)} \tilde{\bm v} + \mathbf{q}^{(n)}\right)^{\sf T} \bm y\\
                \text { s.t. } &~  \tilde{\bm v} \in \mathcal{V}, ~ \bm y \in \mathcal{Y},
        \end{split}
\end{equation}
where $\mathbf{P}^{(n)} = \left[\mathbf{p}_{1}^{(n)},\mathbf{p}_{2}^{(n)},\ldots,\mathbf{p}_{K}^{(n)}\right]^{\sf T}$, $\mathbf{q}^{(n)} = \left[q_{1}^{(n)},q_{2}^{(n)},\ldots,q_{K}^{(n)}\right]^{\sf T}$,
%\[
%\begin{aligned}
%       \begin{split}
%               \mathbf{P}^{(n)} &= \left[\mathbf{p}_{1}^{(n)},\mathbf{p}_{2}^{(n)},\ldots,\mathbf{p}_{K}^{(n)}\right]^{\sf T},\\
%\mathbf{q}^{(n)} &= \left[q_{1}^{(n)},q_{2}^{(n)},\ldots,q_{K}^{(n)}\right]^{\sf T},\\
%       \end{split}
%\end{aligned}
%\]
and $\bm y$ is the Lagrangian dual variable with set 
$\mathcal{Y} = \{\bm y|\bm y \geq \mathbf{0},\mathbf{1}^{\sf T} \bm y=1,\bm y\in\mathcal{R}^{K}\}$ being the feasible domain. 
%Note that $\mathcal{Y}$ is a $K$-dimensional probability simplex.
}
\end{lemma}
\begin{proof}
Please refer to Appendix C.     
\end{proof}

Different from the Nesterov's smoothing technique discussed in Section \ref{discussion}, which relies on the smoothness parameter to balance the tradeoff between approximation accuracy and computation efficiency,
our proposed method is parameter-free and the resulting smooth saddle point problem \eqref{ori:saddle4} is equivalent to the non-smooth problem \eqref{saddle:ori2} rather than an approximation.
By denoting $\psi^{(n)}(\tilde{\bm v}, \bm y):=\left(\mathbf{P}^{(n)} \tilde{\bm v} + \mathbf{q}^{(n)}\right)^{\sf T} \bm y,$
solving non-smooth convex problem \eqref{ori:saddle4} is equivalent to finding a saddle point for the smooth convex-concave function $\psi^{(n)}(\tilde{\bm v}, \bm y)$ under $\mathcal{V}\times\mathcal{Y}$.
%Denote function $\psi^{(n)}(\tilde{\bm v}, \bm y):=\left(\mathbf{P}^{(n)} \tilde{\bm v} + \mathbf{q}^{(n)}\right)^{\sf T} \bm y.$
%\textcolor{blue}{
%It is clear that $\psi^{(n)}(\tilde{\bm v}, \bm y)$ is differentiable. 
%We can further obtain the following lemma.}
%
%\emph{Lemma 2:}
%       $\psi^{(n)}(\tilde{\bm v}, \bm y)$ is $\{0,\max_{k}\{\|\mathbf{p}_{k}^{(n)}\|\},0,\max_{k}\{\|\mathbf{p}_{k}^{(n)}\|\} \}$ - smooth with respect to $\|\cdot\|$ on $\mathcal{V}$ and $\|\cdot\|_1$ on $\mathcal{Y}$ ,
%
%\begin{proof}
%Please refer to appendix C for the proof       
%\end{proof}

%\textcolor{blue}{Delete saddle point to variational inequality}

%$$\underset{\bm x}{\min} \max_{\bm y} ~~ \psi^{(n)}(\bm x, \bm y) = \underset{\bm y}{\max} \min_{\tilde{\bm v}} ~~ \psi^{(n)}(\bm x, \bm y).$$

\subsubsection{First-Order Optimality Condition}
By denoting $\left(\tilde{\bm v}^{*}, \mathbf{y}^{*}\right)$ as the saddle point of $\psi^{(n)}(\tilde{\bm v}, \mathbf{y})$, we have
\begin{equation}
\label{saddle_ineq}
\psi^{(n)} \left(\tilde{\bm v}^{*}, \bm{y}\right) \leq \psi^{(n)}\left(\tilde{\bm v}^{*}, \bm{y}^{*}\right) \leq \psi^{(n)} \left(\tilde{\bm v}, \bm{y}^{*}\right), \forall (\tilde{\bm v},\bm y )\in \mathcal{V} \times \mathcal{Y}. 
\end{equation}
%As $\psi^{(n)}(\tilde{\bm v}, \bm y)$ is differentiable, 
The first-order optimality condition \cite{boyd_vandenberghe_2004} of a saddle point for $\psi^{(n)}(\tilde{\bm v}, \bm y)$ is given by
%\begin{equation}
%\label{zero_order}
%        \begin{aligned}
%             \psi^{(n)}\left(\tilde{\bm v}, \bm{y}^{*}\right) - \psi^{(n)}\left({\tilde{\bm v}}^{*}, \bm{y}^{*}\right) \geq & 0,\\
%                \psi^{(n)}\left(\tilde{\bm v}^{*}, \bm{y}^{*}\right) - \psi^{(n)}\left(\tilde{\bm v}^{*}, \bm{y}\right) \geq & 0.
%        \end{aligned}
%\end{equation}
%Furthermore, \eqref{zero_order} can be equivalently expressed as  
\begin{equation}
\begin{aligned}
      \left  \{ 
      \begin{split}
             \nabla_{\tilde{\bm v}} \psi^{(n)}\left(\tilde{\bm v}^{*}, \bm{y}^{*}\right)(\tilde{\bm v}-\tilde{\bm v}^{*}) \geq & 0,\\
                -\nabla_{\bm{y}} \psi^{(n)}\left(\tilde{\bm v}^{*}, \bm{y}^{*}\right)(\bm{y}-\bm{y}^{*}) \geq & 0,
      \end{split}\right. \quad
      \begin{split}
      \forall ~ (\tilde{\bm v},\bm y )\in \mathcal{V} \times \mathcal{Y}.
      \end{split} 
      \label{two_inequality}
\end{aligned}
\end{equation}
By denoting $\bm{z}=\left[\tilde{\bm v}^{\sf T}, \bm{y}^{\sf T}\right]^{\sf T}$ and
\begin{equation}
\label{compute_F}
\mathbf{F}(\bm{z}):=
\begin{bmatrix}
& \nabla_{\tilde{\bm v}} \psi^{(n)}({\tilde{\bm v}}, \bm{y}) \\
&-\nabla_{\bm{y}}  \psi^{(n)}(\tilde{\bm v}, \bm{y})
\end{bmatrix} = 
\begin{bmatrix}
&(\mathbf{P}^{(n)})^{\sf T}\bm y \\ 
&-\left(\mathbf{P}^{(n)} \tilde{\bm v} + \mathbf{q}^{(n)} \right)     
\end{bmatrix}, \nonumber
\end{equation}
the first-order optimality condition \eqref{two_inequality} can be rewritten in a more compact form as follows
\begin{equation}
\mathbf{F}\left(\bm{z}^{*}\right)^{\sf T}\left(\bm{z}-\bm{z}^{*}\right) \geq 0, \forall ~  \bm{z} \in \mathcal{V}\times\mathcal{Y}.
\end{equation}   
\begin{lemma}
\label{Lipschitz}
\emph{
The operator $\mathbf{F}(\bm{z})$ is monotone and $L$-Lipschitz continuous on space $\mathcal{V}\times\mathcal{Y}$, where $\mathcal{V}$ is endowed with $l_2$ norm, $\mathcal{Y}$ is endowed with $l_1$ norm, and the Lipschitz parameter $L = \max_{k}\{\|\mathbf{p}_{k}^{(n)}\|\}$.
%We denote the Lipschitz parameter as $L = \max_{k}\{\|\mathbf{p}_{k}^{(n)}\|$.
}
\end{lemma}
\begin{proof}
Please refer to Appendix D.     
\end{proof}

Given the above,
problem \eqref{ori:saddle4} is equivalent to the following variational inequality problem with monotone and $L$-Lipschitz continuous operator:
\begin{equation}\label{VI}
\begin{aligned}
\text{find} \quad &\bm{z}^{*} \\
\text{s.t.}\quad &\mathbf{F}\left(\bm{z}^{*}\right)^{\sf T}\left(\bm{z}-\bm{z}^{*}\right) \geq 0, \forall ~  \bm{z} \in \mathcal{V}\times\mathcal{Y}, \\
 &\bm{z}^{*} \in \mathcal{V}\times\mathcal{Y}.
\end{aligned}
\end{equation}

\begin{remark}
\emph{
$\mathbf{F}(\bm{z})$ can be regarded as a gradient-type vector field on space $\mathcal{V} \times \mathcal{Y}$.
The variational inequality $\mathbf{F}\left(\bm{z}^{*}\right)^{\sf T}\left(\bm{z}-\bm{z}^{*}\right) \geq 0$ is similar to the first-order optimality condition for convex constrained problems, where $\mathbf{F}$ resembles the gradient or subgradient.
Hence, an intuitive solution is to employ the generalized projected gradient method \cite{ryu2016primer} to solve problem \eqref{VI}.
 However, the classical gradient-type algorithm cannot monitor the local geometry in the non-Euclidean space \cite{mirror1983}, thereby, weakening the algorithm performance. 
 This motivates the Mirror-Prox method, which we shall present as follows.   
 }
 \end{remark}

\subsubsection{Mirror-Prox Method}

%The authors in \cite{variational} proposed a mirror-prox algorithm which is similar to mirror descent, but Mirror-prox method owns a better convergence property.
%\mathcal{V}\times\mathcal{Y} 
The Mirror-Prox method was firstly proposed in \cite{variational} for solving the Lipschitz continuous variational inequality problem with convergence rate $\mathcal{O}\left(\frac{1}{t} \right)$.
Specifically, at each iteration, the Mirror-Prox method updates $\bm z$ through the following two steps 
\begin{subequations}
\label{mirror-prox-ori}
\begin{align}
\begin{split}
\label{mirror-prox-ori-a}
\bm z_{t+1}^{\prime}&=\underset{\bm z \in \mathcal{V} \times \mathcal{Y}}{\argmin}\left\{D\left(\bm z, \bm z_{t}\right)+\left\langle\gamma \mathbf{F}\left(\bm z_{t}\right), \bm z\right\rangle\right\},
\end{split}\\
\begin{split}
\label{mirror-prox-ori-b}
\bm z_{t+1}&=\underset{\bm z \in \mathcal{V} \times \mathcal{Y}}{\argmin}\left\{D\left(\bm z, \bm z_{t}\right)+\left\langle\gamma \mathbf{F}\left(\bm z_{t+1}^{\prime}\right), \bm z\right\rangle\right\},
\end{split}
\end{align}
\end{subequations}
where $D\left(\bm z, \bm z_{t}\right)$ denotes the Bregman distance and $\gamma$ is a parameter that is determined by the Lipschitz parameter of $\mathbf{F}(\bm{z})$.
%to attain the ergodic convergence rate $\mathcal{O}\left(\frac{1}{t}\right)$, the step-size $\gamma$ should be smaller than $1/\mu$.    
According to the analysis in \cite{variational,Bubeck2015}, we set $\gamma = 1/\left(2\max_{k}\{\|\mathbf{p}_{k}^{(n)}\|\}\right)$ to achieve the desired convergence rate.
Besides, the use of Bregman distance $D\left(\bm z, \bm z_{t}\right)$ is to monitor the local geometry for improving the algorithm performance \cite{mirror1983}.
Specifically, Bregman distance $D\left(\bm z, \bm z_{t}\right)$ is induced by a mirror mapping function $\Phi(\bm z):~\mathbb{R}^{2N}\times\mathbb{R}^{K} \rightarrow \mathbb{R}$ for set $\mathcal{V}\times\mathcal{Y}$ :
\begin{equation}
\begin{split}
        D\left(\bm z, \bm z_{t}\right) &= \Phi(\bm z) - \Phi(\bm z_t) - \left\langle \nabla \Phi(\bm z_t), \bm z-\bm z_t \right\rangle. \label{bregman divergence}
\end{split}     
\end{equation}
We select the mapping function according to the structure of $\mathcal{V}\times\mathcal{Y}$ to attain the goal of capturing the local geometry.
In our work, $\Phi(\bm z):~\mathbb{R}^{2N}\times\mathbb{R}^{K} \rightarrow \mathbb{R}$ is defined as
\[
\Phi(\bm z) =  \frac{1}{2}\|\tilde{\bm v}\|^{2} + \sum_{k=1}^{K} \bm{y}^{k} \log \bm{y}^{k},
\]
where the first item and the second item are the generally used mirror mapping function for an Euclidean space, like $\mathcal{V}$, and a simplex space, like $\mathcal{Y}$, respectively.  
As a result, we have
\begin{equation}
\begin{split}
        D\left(\bm z, \bm z_{t}\right) = \frac{1}{2}\|\tilde{\bm v}
-\tilde{\bm v}_t\|^{2}+\sum_{k=1}^{K} y^{k} \log \frac{y^{k}}{y_t^{k}}
-\sum_{k=1}^{K}(y^{k}-y_t^{k}).\label{bregman divergence explicit}
\end{split}     
\end{equation}

%\begin{equation}
%\begin{split}
%       D\left(\bm z, \bm z_{t}\right) &= \Phi(\bm z) - \Phi(\bm z_t) - \left\langle \nabla \Phi(\bm z_t), \bm z-\bm z_t \right\rangle \label{bregman divergence} \\
%       &= \frac{1}{2}\|\tilde{\bm v}
%-\tilde{\bm v}_t\|^{2}+\sum_{k=1}^{K} y^{k} \log \frac{y^{k}}{y_t^{k}}
%-\sum_{k=1}^{K}(y^{k}-y_t^{k}).\label{bregman divergence explicit}
%\end{split}    
%\end{equation}

\begin{remark}
\emph{
The update of the Mirror-Prox method has an additional Step \eqref{mirror-prox-ori-b} compared to the mirror descent method that only requires Step \eqref{mirror-prox-ori-a}. 
The computation of $\bm z_{t+1}^{\prime}$ in \eqref{mirror-prox-ori-a} is used to find a better direction $\mathbf{F}\left(\bm z_{t+1}^{\prime}\right)$ than the mirror descent method
which is further used to update $\bm z_{t+1}$, i.e., \eqref{mirror-prox-ori-b}.  
This process is illustrated in Fig. \ref{mirror}.
%\vspace{-3mm}
\begin{figure}[t]
        \centering
                \includegraphics[width=5cm,height=3.5cm]{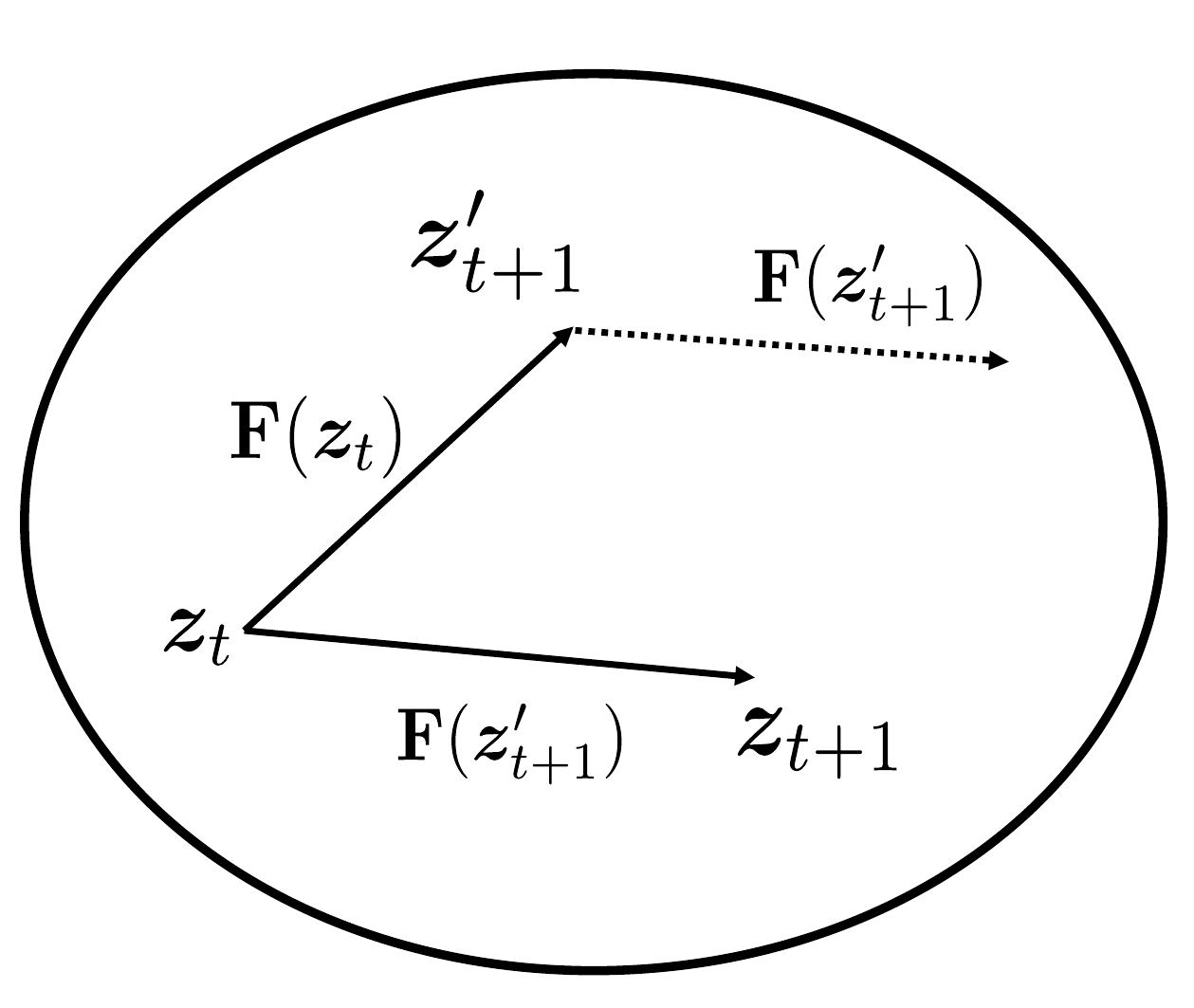}
                \vspace{-3mm}
                \caption{Illustration of the two steps of the Mirror-Prox method in the $t$-th iteration.}\label{mirror}
                \vspace{-6mm}
\end{figure}
The Mirror-Prox method improves the convergence rate from $\mathcal{O}\left(\frac{1}{\sqrt{t}} \right)$ achieved by mirror descent to $\mathcal{O}\left(\frac{1}{t} \right)$ by using a more accurate update. 
We will show that the updates of \eqref{mirror-prox-ori-a} and \eqref{mirror-prox-ori-b} are of low complexity. 
}
\end{remark}

\subsubsection{Implementation Details}

At each iteration of the Mirror-Prox method, we need to efficiently solve two optimization problems \eqref{mirror-prox-ori-a} and \eqref{mirror-prox-ori-b}. 
Specifically, by introducing a variable $\bm w$, the update of $\bm z_{t+1}^{\prime}$ in \eqref{mirror-prox-ori-a} can be decomposed into the following three steps \cite{Bubeck2015}
%\vspace{-8mm}
\begin{subequations}
\label{upd}
\begin{align}
        \nabla \Phi(\bm w) &= \nabla \Phi(\bm z_t) - \gamma \mathbf{F}\left(\bm z_{t}\right),\label{aaa}  \\
        \bm w ~ &=  \nabla \Phi^{-1}(\nabla \Phi(\bm z_t) - \gamma \mathbf{F}\left(\bm z_{t}\right)), \label{aa} \\
        \bm z_{t+1}^{\prime} ~ &= \underset{\bm z \in \mathcal{V} \times \mathcal{Y}}{\operatorname{arg} \min}\left\{D(\bm z, \bm w)\right\}. \label{a}
\end{align}     
\end{subequations}
It is clear that \eqref{aaa} is easy to compute.  
Furthermore, $\bm w$ in \eqref{aa} can be obtained in a closed form, since $\nabla \Phi(\cdot)$ and $\nabla \Phi^{-1} (\cdot)$ can analytically be expressed as    
\begin{equation}
\nabla \Phi({\mathbf{z}})=
\begin{bmatrix}
\tilde{\bm v}\\
1+\log {y}_{1} \\
\vdots \\
1+\log {y}_{K}
\end{bmatrix},
\nabla \Phi^{-1}(\bm \xi)=
\begin{bmatrix}
\bm \zeta \\
\text{exp}(\nu_1 - 1)\\
\text{exp}(\nu_2 - 1)\\
\vdots \\
\text{exp}(\nu_K - 1)
\end{bmatrix}, \nonumber
\end{equation}
where $\bm \xi = [\bm \zeta^{\sf T}, \bm \nu ^{\sf T}]^{\sf T}$.
Subsequently, we show that the solution of problem \eqref{a} also admits a closed form, which involves a projection problem on $\mathcal{V} \times \mathcal{Y}$.
According to \eqref{bregman divergence explicit}, we rewrite problem \eqref{a} as
\begin{align}
\label{two_projection}
\begin{autobreak}
\underset{\bm z \in \mathcal{V} \times \mathcal{Y}}{\argmin} \{D(\bm z, \bm w)\}=
\underset{\bm z \in \mathcal{V} \times \mathcal{Y}}{\argmin} \Big \{\frac{1}{2}\|\tilde{\bm v}
-\bm{u}\|^{2}+\sum_{k=1}^{K} y^{k} \log \frac{y^{k}}{e^{k}}
-\sum_{k=1}^{K}(y^{k}-e^{k})\Big\},
\end{autobreak}
\end{align}
where $\bm w = [\bm{u}^{\sf T}, \bm{e}^{\sf T}]^{\sf T}$.
Problem \eqref{two_projection} can be decomposed into two independent subproblems with respect to $\tilde{\bm v}$ and $\bm y$, respectively, given by
\begin{subequations}
\begin{align}
        \tilde{\bm v}_{t+1}^{\prime} =\underset{\tilde{\bm v} \in \mathcal{V}}{\argmin}&\left\{\frac{1}{2}\left\|\tilde{\bm v}-\bm{u}\right\|^{2}\right\}, \label{x} \\
        \bm y_{t+1}^{\prime} =\underset{\bm y \in \mathcal{Y}}{\argmin}&\left\{\sum_{k=1}^{K}{y}^{k} \log \frac{{y}^{k}}{e^{k}}-\sum_{k=1}^{K}\left({y}^{k}-e^{k}\right)\right\}.\label{y}
\end{align}     
\end{subequations}
Problem \eqref{x} can be considered as a projection problem in an Euclidean space \cite{boyd_vandenberghe_2004}.
%, i.e., $\tilde{\bm v} = \operatorname{Proj}_{\mathcal{V}}\left(\mathbf{u}\right)$
The optimal $(\tilde{v}_{t+1})_i,~\forall ~  i $, admit the following expression
\begin{equation}
\label{compute_v}
\begin{split}
        (\tilde{v}_{t+1})_i = \left\{\begin{array}{cl}
\frac{u_i}{\left(u_i^{2} + u_{N+i}^{2}\right)^{\frac{1}{2}}}, & u_i^{2} + u_{N+i}^{2} \ge 1, \\
u_i, & \text { otherwise. }
\end{array}\right.
\end{split}     
\end{equation}
Problem \eqref{y} is a projection problem in a simplex space \cite{Bubeck2015}.
%, i.e., , $\tilde{\bm y} = \operatorname{Proj}_{\mathcal{Y}}\left(\mathbf{1}\right)$
The optimal $\bm y$ is given by
\begin{equation}
\label{compute_y}
\begin{split}
        \bm{y}=\left\{\begin{array}{cl}
\bm{e}, & \bm{e} \in \mathcal{Y}, \\
\frac{\bm{e}}{\left\|\bm{e}\right\|_{1}}, & \text { otherwise. }
\end{array}\right.
\end{split}     
\end{equation}
%For ease of notations, we denote the projection of \eqref{compute_v} and \eqref{compute_y} as $\operatorname{Proj}_{\mathcal{V} \times \mathcal{Y}}\left(\bm w\right)$, i.e., $\operatorname{Proj}_{\mathcal{V} \times \mathcal{Y}}\left(\bm w\right) = \underset{\bm z \in \mathcal{V} \times \mathcal{Y}}{\argmin} \{D(\bm z, \bm w)\}$. 
The above derivation for problem \eqref{mirror-prox-ori-a} can be directly applied to problem \eqref{mirror-prox-ori-b}.
%Similarly, denote $\nabla G(\bm w) = \nabla G(\bm z_t) - \gamma \mathbf{F}\left(\bm z_{t+1}^{\prime}\right),$
%problem \eqref{mirror-prox-ori-b} can be decomposed into 
%\begin{subequations}
%\begin{align}
%       &\nabla G(\bm w) =\nabla G(\bm z_t) - \gamma \mathbf{F}\left(\bm z_{t+1}^{\prime}\right) \label{bbb}  \\
%       &\bm w =  \nabla G^{-1}(\nabla G(\bm z_t) - \gamma \mathbf{F}\left(\bm z_{t+1}^{\prime}\right)) \label{bb} \\
%       &\bm z_{t+1} = \underset{\bm z \in \mathcal{V} \times \mathcal{Y}}{\operatorname{arg} \min}\left\{D(\bm z, \bm w)\right\} \label{b}
%\end{align}    
%\end{subequations}
The details of the proposed Mirror-Prox algorithm for solving problem \eqref{VI} is summarized in Algorithm 2.

\begin{algorithm}
\label{alg:general}
\caption{Mirror-Prox method for Problem \eqref{saddle:ori2}}
\SetKwRepeat{Repeat}{repeat}{until}
%\SetKwFor{For}{For}{end}
\SetAlgoLined
\KwIn{Initial point $\bm z_0 =\left[\tilde{\bm v}_0^{\sf T}, \bm{y}_0^{\sf T}\right]^{\sf T}$, threshold $\epsilon$, and stepsize
$\gamma = \frac{1}{2\max_{k}\{\|\mathbf{p}_{k}^{(n)}\|\}}$;}
\For{$t = 1, 2, \ldots, $}{
$\mathbf{F}(\bm{z}_t)=
\begin{bmatrix}
        \left((\mathbf{P}^{(n)})^{\sf T}\bm y_t \right)^{\sf T} ,
        -\left(\mathbf{P}^{(n)} \tilde{\bm v}_t + \mathbf{q}^{(n)}\right)^{\sf T}
\end{bmatrix}^{\sf T}$;\\
$\nabla \Phi(\bm w) = \nabla \Phi(\bm z_t) - \gamma \mathbf{F}\left(\bm z_{t}\right)$; \\
$\bm w =  \nabla \Phi^{-1}(\nabla \Phi(\bm z_t) - \gamma \mathbf{F}\left(\bm z_{t}\right))$; \\
$\bm z_{t+1}^{\prime} = \underset{\bm z \in \mathcal{V} \times \mathcal{Y}}{\argmin} \{D(\bm z, \bm w)\}$; \\
%Project $\bm w$ to set $\mathcal{V} \times \mathcal{Y}$ according to \eqref{compute_v} and \eqref{compute_y};
$\nabla \Phi(\bm w) = \nabla \Phi(\bm z_t) - \gamma \mathbf{F}\left(\bm z_{t+1}^{\prime}\right)$;\\
$\bm w = \nabla \Phi^{-1}(\nabla \Phi(\bm z_t) - \gamma \mathbf{F}\left(\bm z_{t+1}^{\prime}\right))$;\\
%Project $\bm w$ to set $\mathcal{V} \times \mathcal{Y}$ according to \eqref{compute_v} and \eqref{compute_y};
$\bm z_{t+1} = \underset{\bm z \in \mathcal{V} \times \mathcal{Y}}{\argmin} \{D(\bm z, \bm w)\}$;\\
\textbf{If} $D(\bm z_t, \bm z_{t+1})$ less than $\epsilon$, set $\bm z_{*} = \frac{1}{T} \sum_{t=1}^{T} \bm{z}_{t} $, \textbf{break}; \textbf{else go to Step 1};
        } 
%       \KwOut{Take $\tilde{\bm v}^{(n+1)}$ as an approximation solution.}
\end{algorithm}

\subsection{Mirror-Prox Method for Non-smooth Convex Problem \eqref{SCA_m}}

The method proposed for solving problem \eqref{saddle:ori2} can be readily applied to solve problem \eqref{SCA_m}.
We next sketch the process of transforming problem \eqref{SCA_m} to its equivalent variational inequality problem. 
%and then present the details of Mirror-Prox method for solving variational inequality problem.
According to Lemma \ref{P_D}, problem \eqref{SCA_m} is equivalent to the following problem 
 \begin{equation}\label{ori:saddle44}
        \begin{split}
             \underset{\tilde{\bm m}}{\min} &~ \max_{\bm y} ~ \left(\mathbf{\bar P}^{(n)} \tilde{\bm m} + \mathbf{\bar q}^{(n)}\right)^{\sf T} \bm y\\
                \text { s.t. } &~  \tilde{\bm m} \in \mathcal{M},~ \bm y \in \mathcal{Y},
        \end{split}
\end{equation}
where $\mathbf{\bar P}^{(n)} = \left[\mathbf{\bar p}_{1}^{(n)},\ldots,\mathbf{\bar p}_{K}^{(n)}\right]^{\sf T}$, 
$\mathbf{\bar q}^{(n)} = \left[\bar q_{1}^{(n)},\ldots,\bar q_{K}^{(n)}\right]^{\sf T}$.
%\[
%\begin{aligned}
%       \begin{split}
%               \mathbf{\bar P}^{(n)} &= \left[\mathbf{\bar p}_{1}^{(n)},\mathbf{\bar p}_{2}^{(n)},\ldots,\mathbf{\bar p}_{K}^{(n)}\right]^{\sf T},\\
%\mathbf{\bar q}^{(n)} &= \left[\bar q_{1}^{(n)},\bar q_{2}^{(n)},\ldots,\bar q_{K}^{(n)}\right]^{\sf T}.\\
%       \end{split}
%\end{aligned}
%\]
By further denoting $\bm{\bar z}=\left[\tilde{\bm m}^{\sf T}, \bm{y}^{\sf T}\right]^{\sf T}$, $\varphi^{(n)}({\tilde{\bm m}}, \bm{y}) = \left(\mathbf{\bar P}^{(n)} \tilde{\bm m} + \mathbf{\bar q}^{(n)}\right)^{\sf T} \bm y$, and 
\[
\mathbf{\bar F}(\bm{\bar z}):=
\begin{bmatrix}
& \nabla_{\tilde{\bm m}} \varphi^{(n)}({\tilde{\bm m}}, \bm{y}) \\
&-\nabla_{\bm{y}}  \varphi^{(n)}(\tilde{\bm m}, \bm{y})
\end{bmatrix} 
=
\begin{bmatrix}
        (\mathbf{\bar P}^{(n)})^{\sf T}\bm y \\
      - \left( \mathbf{\bar P}^{(n)} \tilde{\bm m} + \mathbf{\bar q}^{(n)} \right)
\end{bmatrix}.
\]
We denote the Lipschitz parameter of operator $\mathbf{\bar F}(\bm{\bar z})$ as $\bar{L}$.
According to Lemma \ref{Lipschitz}, we have $\bar{L} = \max_{k}\{\|\bar{\mathbf{p}}_{k}^{(n)}\|\}$.
Problem \eqref{ori:saddle44} can be further transformed to the following variational inequality problem
\begin{equation}\label{VI_m}
\begin{aligned}
\text{find} \quad &\bm{\bar z}^{*} \\
\text{s.t.}\quad &\mathbf{\bar F}\left(\bm{\bar z}^{*}\right)^{\sf T}\left(\bm{\bar z}-\bm{\bar z}^{*}\right) \geq 0, \forall ~  \bm{\bar z} \in \mathcal{M}\times\mathcal{Y}, \\
 &\bm{\bar z}^{*} \in \mathcal{M}\times\mathcal{Y}.
\end{aligned}
\end{equation}
Algorithm 2 can be applied to solve problem \eqref{VI_m} by replacing $\bm{z}$, $\mathbf{p}_{k}^{(n)}$, and $\mathbf{F}(\bm{z}_t)$ by $\bm{\bar{z}}$, $\mathbf{\bar p}_{k}^{(n)}$, and $\mathbf{\bar F}(\bm{\bar z}_t)$, respectively.

\subsection{Computational Complexity}
At each iteration of Algorithm 2, the computation complexity is dominated by Step 2.
In Step 2, we take a matrix-vector multiplication to update operator $\mathbf{F}(\cdot)$.
For the optimization of phase-shift vector $\tilde{\bm v} \in \mathbb{R}^{2N}$,
the computation cost of operator $\mathbf{F}(\cdot)$ in Step 2 is $\mathcal{O}(NK)$.
Besides, according to \cite{Bubeck2015}, Algorithm 2 can obtain an $\epsilon$-optimal solution within $\mathcal{O}\left(\frac{L \log (N)}{\epsilon}\right)$ iterations.
As a result, the computation complexity of Algorithm 2 is $\mathcal{O}\left(\frac{NKL \log (N)}{\epsilon}\right)$.
In a similar manner, we can conclude that the computation complexity is $\mathcal{O}\left(\frac{MK \bar{L} \log (M)}{\epsilon}\right)$ for updating $\tilde{\bm m}$.
The specific time savings in our problem achieved by the proposed algorithm will be further demonstrated in the following section via simulations.

\section{Simulation Results}
In this section, we present sample simulation results to illustrate the performance of the proposed algorithm for minimizing the MSE of the RIS-assisted AirComp systems. 

\subsection{Simulation Settings}
We consider a  three-dimensional (3D) coordinate system. The AP and the RIS are, respectively, located at $(0,~ 0,~ 20)$ meters and $ (100,~ 0,~ 20) $ meters, while the IoT devices are uniformly located within a circular region centered at $(100,~ 20,~ 0)$ meters with radius $20$ meters. 
The antennas at the AP and the passive reflecting elements at the RIS are arranged as a uniform linear array and a uniform planar array, respectively. 
In the simulations, we consider both large-scale fading and small-scale fading for all the channels.
The distance-dependent large-scale fading is modeled as $T_0(d/d_0)^{-\alpha}$, 
where $T_0$ is the path loss at the reference distance $d_0 = 1$ meter, $d$ denotes the distance between the transmitter and the receiver, and $\alpha$ is the path loss exponent.
We consider Rayleigh fading for the direct channel.
%Hence, the channel coefficient of the direct link from IoT device $k$ to the AP is
%\[
%\bm h_{\mathrm{d},k}=\sqrt{T_0(d_{\mathrm{DA}}/d_0)^{-\alpha_{\mathrm{DA},k}}}\bm h_{\mathrm{d},k}^\mathrm{NLoS}, ~\forall k
%\]
%where $\bm h_{\mathrm{d},k}^\mathrm{NLoS}$ denotes the non-line-of-sight (NLoS) component,
%the distance between IoT devices $k$ and AP is denoted by $d_{\mathrm{DA},k}$ and $\alpha_{\mathrm{DA}}$ is the path loss.
Besides, we consider Rician fading for the reflecting links with Rician factor $\beta$.
For the RIS-AP link (i.e., $\bm G$), we have
\[
\bm G=\sqrt{T_0(d_{\mathrm{RA}}/d_0)^{-\alpha_{\mathrm{RA}}}}\left(\sqrt{\frac{\beta}{1+\beta}}\bm G^\mathrm{LoS}+\sqrt{\frac{1}{1+\beta}}\bm G^\mathrm{NLoS}\right),
\]
%\begin{equation}
%\mathbf{G}=\sqrt{T_0d_{IA}^{-\alpha_{IA}}}\left(\sqrt{\frac{\beta_{IA}}{1+\beta_{IA}}}\mathbf{G}^\mathrm{LoS}+\sqrt{\frac{1}{1+\beta_{IA}}}\mathbf{G}^\mathrm{NLoS}\right),\nonumber
%\end{equation}
where
$\bm{G}^\mathrm{LoS}$ and $\bm G^\mathrm{NLoS}$ denote the line-of-sight (LoS) and the non-line-of-sight (NLoS) components, respectively,
$d_{\mathrm{RA}}$ is the distance between the RIS and the AP, and $\alpha_{\mathrm{RA}}$ is the path loss exponent.
The channel coefficient of the link between IoT device $k$ and the RIS, i.e., $\bm h_{\mathrm{r},k}~ \forall ~ k$, is generated in a similar manner as $\bm G$.
We set the path loss exponents of the device-AP links, the device-RIS links, and the RIS-AP link as $3.8$, $2.5$, and $2.2$, respectively. 
Unless specified otherwise, we set $\beta = 3$, $T_0 = -30 $ dB, $P = 30$ dBm, $\sigma^2 = -90$ dBm, and $\epsilon = 10^{-5}$.

\subsection{Performance Evaluation}
We investigate the convergence performance of the proposed Mirror-Prox based AlterMin SCA algorithm in Fig. \ref{convergence_m_v}.
It can be observed that the MSE of our proposed algorithm monotonically decreases over the iterations and converges in a few iterations. 
In addition, the achieved MSE values of the proposed algorithms before and after the projection of the phase-shift vector at the RIS are almost the same. 
This is because, in the simulations, the obtained phase-shift vector always meets the constraints $|v_i| = 1, \forall ~ i$ before projection. 

\begin{figure}[t]
                \centering
                \includegraphics[width=8cm,height=6cm]{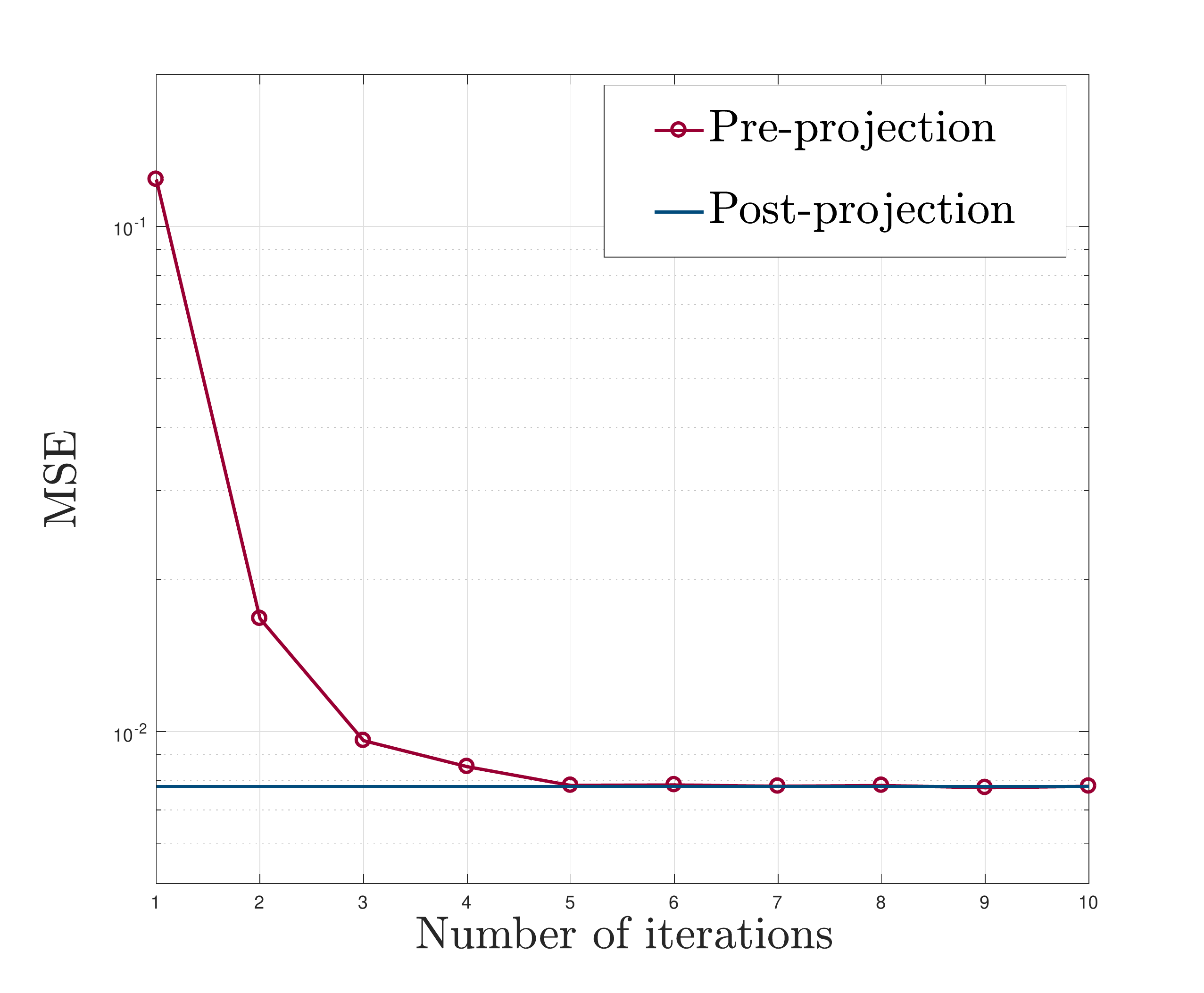}
                \vspace{-3mm}
                \caption{Convergence behavior of the proposed algorithm when $N = 50$, $M = 10$, and $K = 200$.}\label{convergence_m_v}
                \vspace{-5mm}
\end{figure}

We compare the proposed algorithm with the following four baseline methods. 
 \begin{itemize}
 \item  \textbf{Alternating SDR}: This method leverages the SDR technique \cite{Tom_luo_sdr,wu2019intelligent,zhibin} to optimize $\bm m$ and $\bm v$ alternately. The Gaussian randomization technique is applied when the solution returned by the relaxed SDP problem does not meet the rank-one constraint. The number of randomization is set to be $50$.
 \item  \textbf{Alternating DC}: This method was proposed in \cite{jiang2019}, which reformulates the rank-one constrained SDP problem as a DC programming problem, followed by using SCA to obtain the rank-one solution via successively solving the convex approximation of the DC problem. 
 Our comparison with this algorithm is only conducted when the number of IoT devices is small, i.e., low-density scenario, due to its high computational complexity.
 \item  \textbf{Random phase shift}: With this method, the phase-shift matrix $\bm \Theta$ is randomly chosen and kept fixed when optimizing the receive beamforming vector $\bm m$ via our proposed algorithm.
% It worth noting that our proposed algorithm is reduced to Mirror-Prox based SCA algorithm since the alternating procedure can be dropped.    
 \item  \textbf{Without RIS}: In the method, the signals are transmitted only through the direct links, i.e., $\bm \Theta = \bm 0$. We only optimize the receive beamforming vector $\bm m$ via our proposed algorithm. 
\end{itemize}
 
 \begin{figure}[t]
        \centering
        \begin{minipage}{.46\textwidth}
                \centering
                \includegraphics[width=7.5cm,height=6cm]{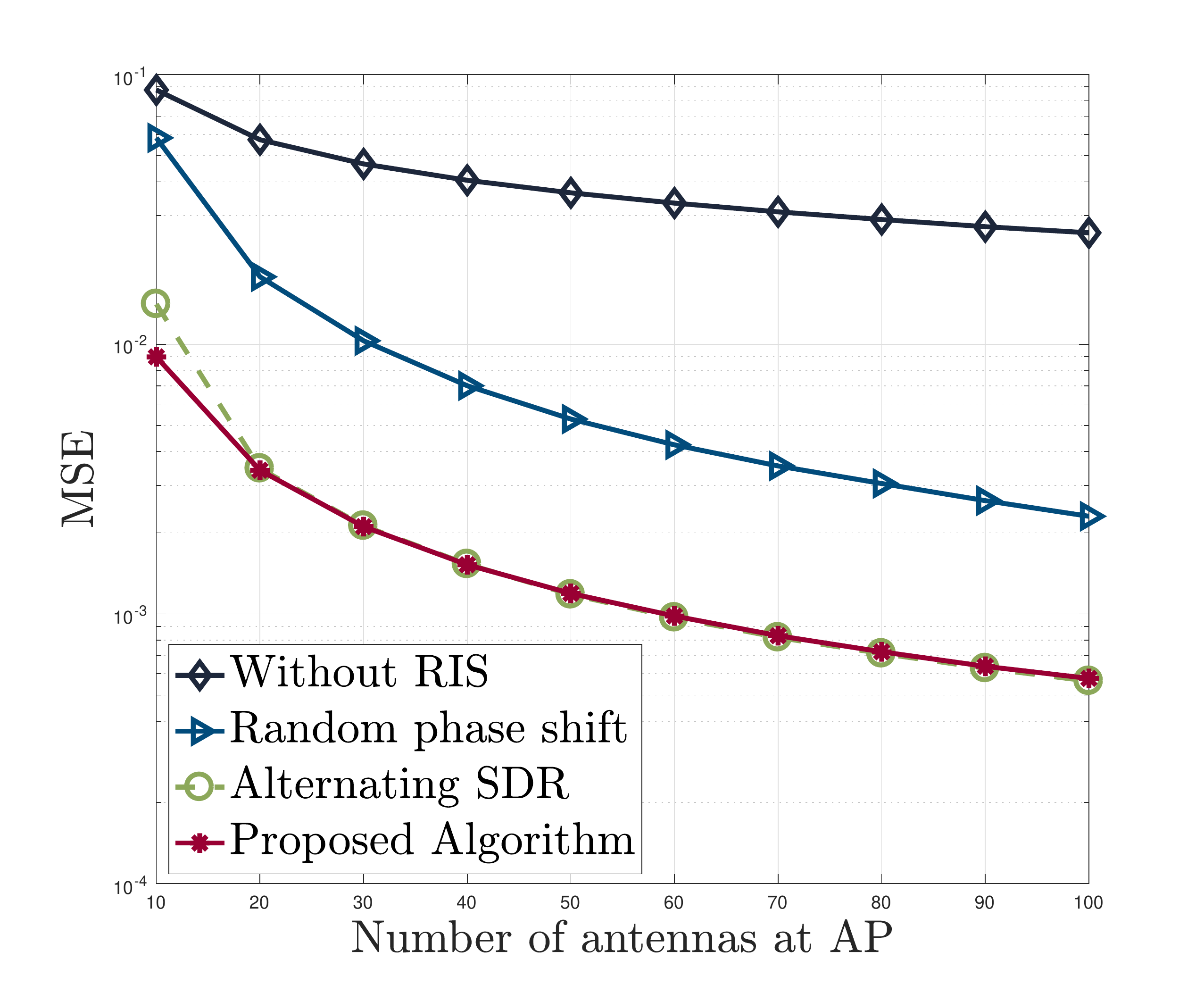}
                \vspace{-9mm}
                \caption{MSE versus the number of antennas at AP when $K=200$ and $N=50$.}\label{l:N}
        \end{minipage}
        \begin{minipage}{.46\textwidth}
                \centering
                \includegraphics[width=7.5cm,height=6cm]{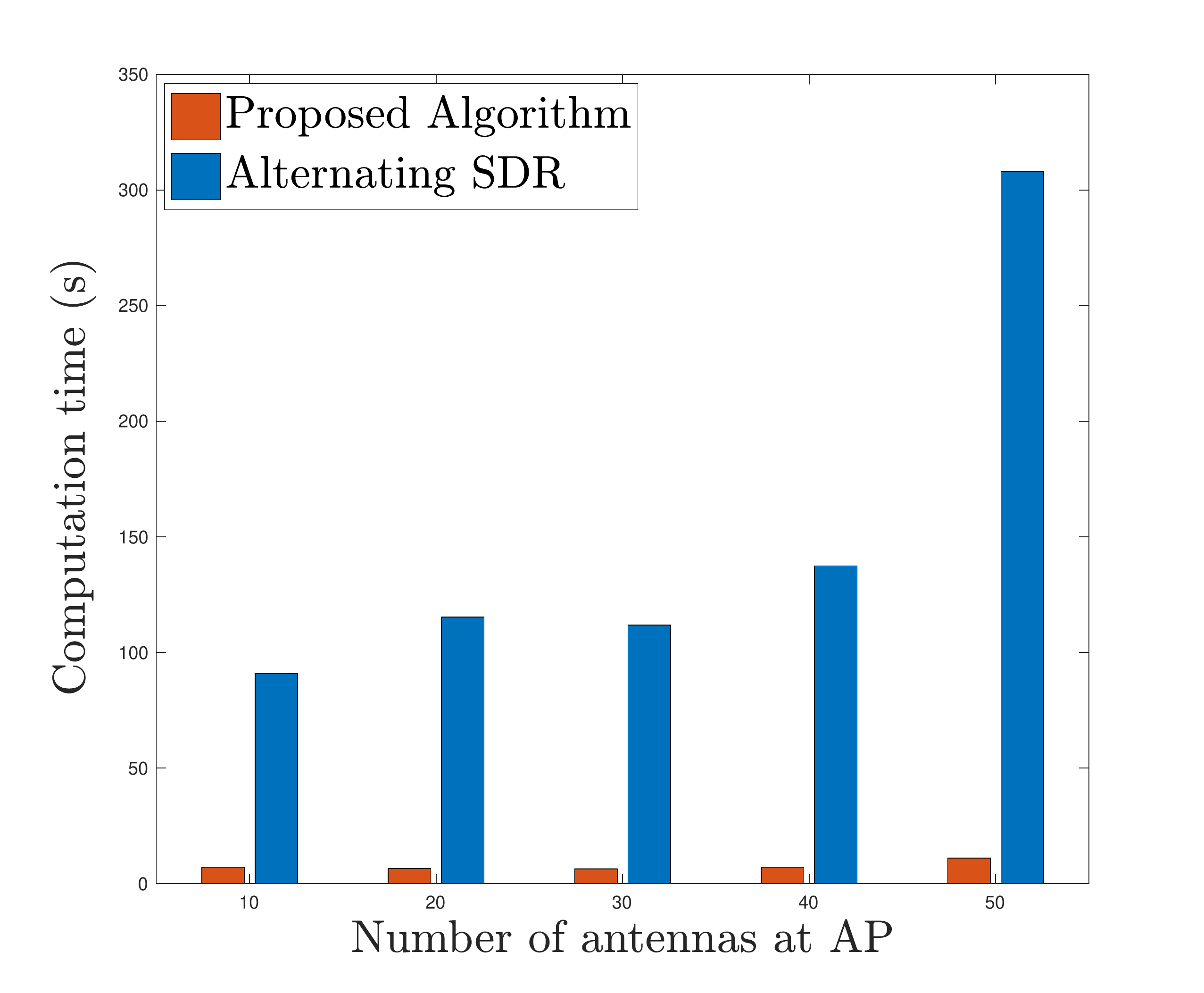}
                \vspace{-9mm}
                \caption{Computation time versus the number of antennas at AP when $K=200$ and $N=50$.}\label{l:t:N}
        \end{minipage}
        \vspace{-5mm}
\end{figure}
%\begin{figure}[t]
%       \centering
%       \subfigure[MSE versus the number of the antennas at AP]{
%               \label{fig_users}
%               \begin{minipage}{0.45\textwidth}
%                       \centering
%                       \includegraphics[scale=0.35]{./figure322/L_AP}
%                       %\setlength{\abovecaptionskip}{2pt}
%               \end{minipage}
%       }
%       \subfigure[Computation time versus the antennas at AP ]{
%               \label{fig_antennas}
%               \begin{minipage}{0.45\textwidth}
%                       \centering
%                       \includegraphics[scale=0.35]{./figure322/L_AP_time}
%                       %\setlength{\abovecaptionskip}{2pt}
%               \end{minipage}
%       } 
%       \vspace{-2mm}
%       \caption{Performance comparison between the proposed algorithm and the baseline methods.}
%       \label{fig_simulation}
%       \vspace{-6mm}
%\end{figure}

%All the figures below are plotted in two scenarios, i.e, small scale and large scale in terms of the number of users.    
%In small-scale Aircomp system which means that the number of the users, i.e, $K$ is not so large.
%

%        \begin{minipage}{.46\textwidth}
%                \centering
%                \includegraphics[width=7.5cm,height=6cm]{./figure322/L_user_time}
%                \vspace{-6mm}
%                \caption{Computation time versus the number of IoT devices when $M=20$ and $N=50$.}\label{l:t:U}
%                \vspace{-9mm}
%        \end{minipage}

%All the simulation results below are obtained by averaging over 500 channel realizations.
As the alternating SDR and alternating DC algorithms are not guaranteed to converge, we set their iteration numbers as the number of iterations required by our proposed algorithm to converge for a fair comparison. 
Based on the number of IoT devices in the network coverage area, we consider the high-density and low-density scenarios. 

\subsubsection{High-density Scenario}
Fig. \ref{l:N} shows the impact of the number of antennas at the AP (i.e., $M$) on the MSE when $N=50$ and $K=200$.  
As can be observed, for all methods under consideration, the MSE of the RIS-assisted AirComp system monotonically decreases as the number of antennas increases.
This is because a greater diversity gain can be achieved by deploying a larger antenna array.
We can also observe that deploying an RIS significantly reduces the MSE in the considered AirComp system.
It shows that RIS is a promising technique that can enhance the performance of AirComp.
Besides, we observe that the proposed algorithm outperforms the benchmark scheme with random phase-shift at the RIS, which reveals the necessity of optimizing the phase-shifts vector in RIS-assisted AirComp systems.
Moreover, our proposed algorithm attains a similar performance as the alternating SDR algorithm.
On the other hand, as can be observed in Fig. \ref{l:t:N}, our proposed algorithm significantly outperforms the alternating SDR algorithm in terms of the computation time. 
Besides, the advantage on the computation time of the proposed algorithm increases as the number of antennas at the AP becomes large.
The reason is that the alternating SDR algorithm requires the execution of the second-order interior point method at each iteration, while the proposed algorithm as a first-order algorithm has a much lower computation complexity than the alternating SDR algorithm.

 \begin{figure}[t]
        \centering
        \begin{minipage}{.46\textwidth}
                \centering
                \includegraphics[width=7.5cm,height=6cm]{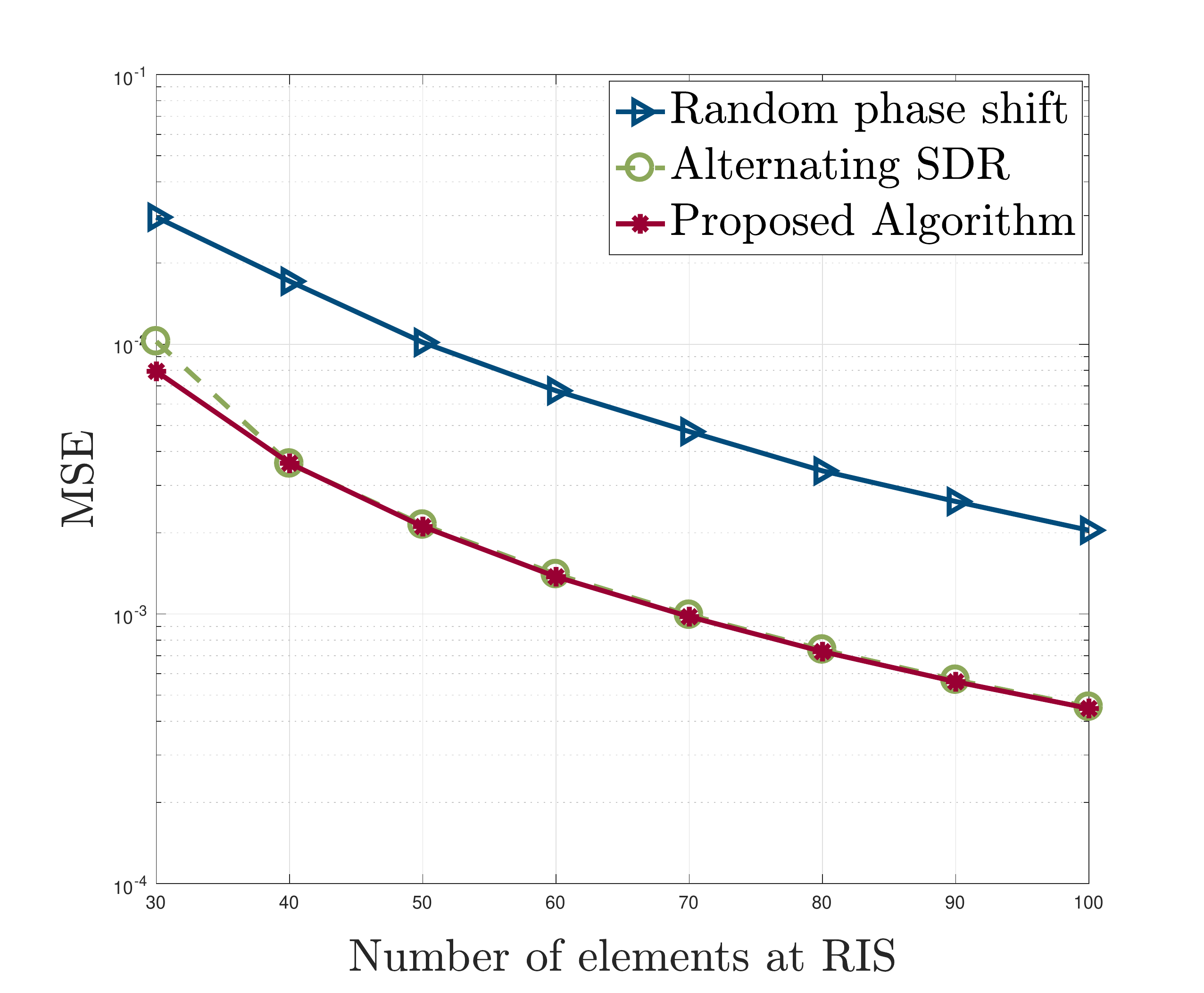}
                \vspace{-9mm}
                \caption{MSE versus the number of reflecting elements at RIS when $K=200$ and $M=20$.}\label{l:M}
                \vspace{-6mm}
        \end{minipage}
        \begin{minipage}{.46\textwidth}
                \centering
                \includegraphics[width=7.5cm,height=6cm]{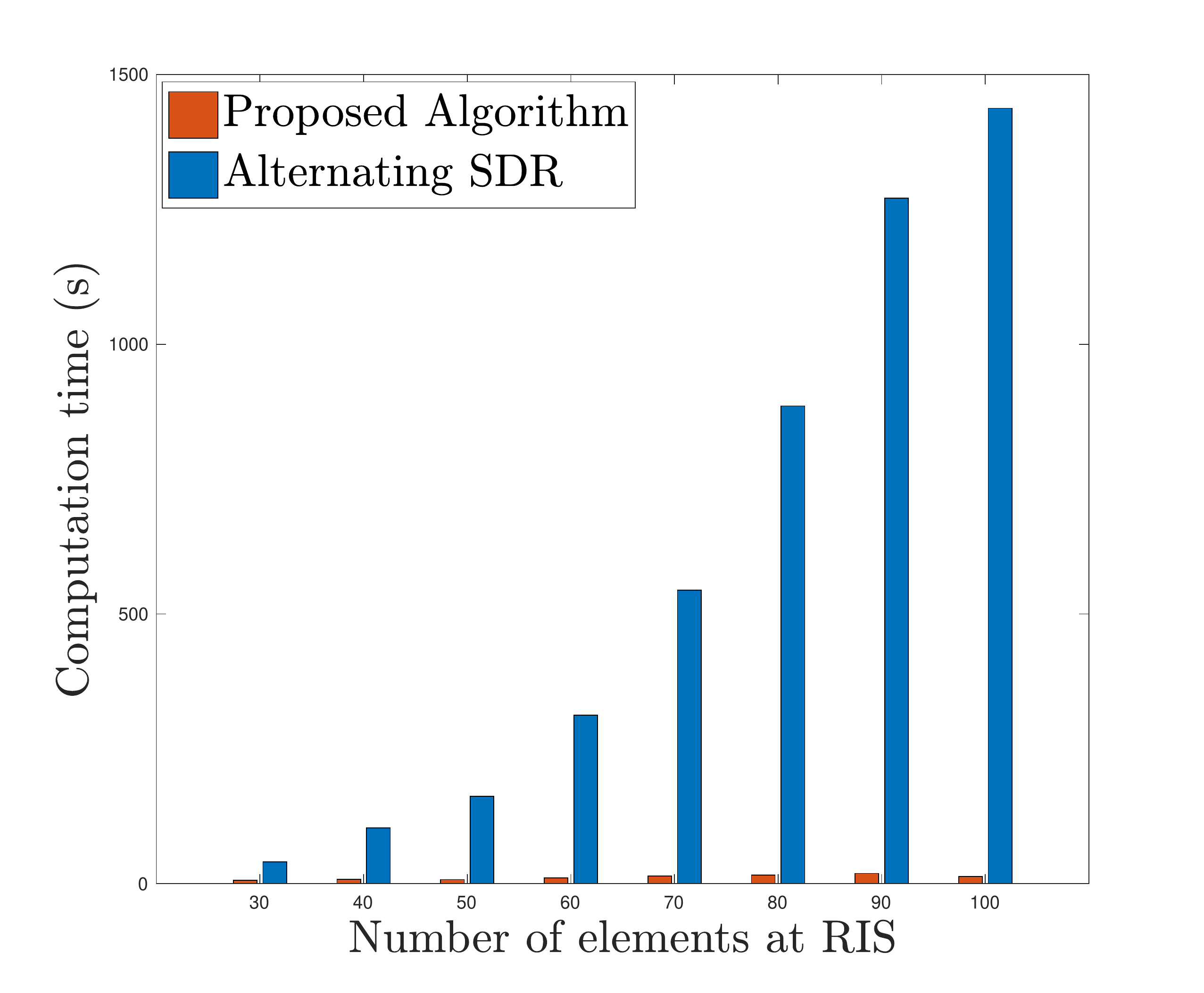}
                \vspace{-9mm}
                \caption{Computation time versus the number of reflecting elements at RIS when $K=200$ and $M=20$.}\label{l:t:M}
                \vspace{-6mm}
        \end{minipage}
\end{figure}

In Fig. \ref{l:M}, we investigate the impact of the number of the reflecting elements at the RIS on MSE when the number of IoT devices $K=200$ and the number of antennas at the AP $M=20$.
As can be observed from Fig. \ref{l:M}, the proposed algorithm achieves almost the same MSE performance for different number of reflecting elements against the benchmark, which confirms the capability of our proposed algorithm to achieve high accurate AirComp. 
In addition, the MSE decreases significantly as the number of reflecting elements at the RIS increases. 
This is due to the fact that the RIS with more reflecting elements has more freedom for the reflection coefficient design.
The computation time of our proposed algorithm and the alternating SDR algorithm versus the number of reflecting elements at the RIS is plotted in Fig. \ref{l:t:M}.
As the number of reflecting elements at the RIS increases, the computation time of the alternating SDR algorithm increases significantly. 
Comparing with the alternating SDR algorithm, the proposed algorithm only consumes $1\%$ and $16\%$ of the computation time when $N=100$ and $N=30$, respectively. 

    \begin{figure}[t]
    \centering
                \includegraphics[width=7.5cm,height=6cm]{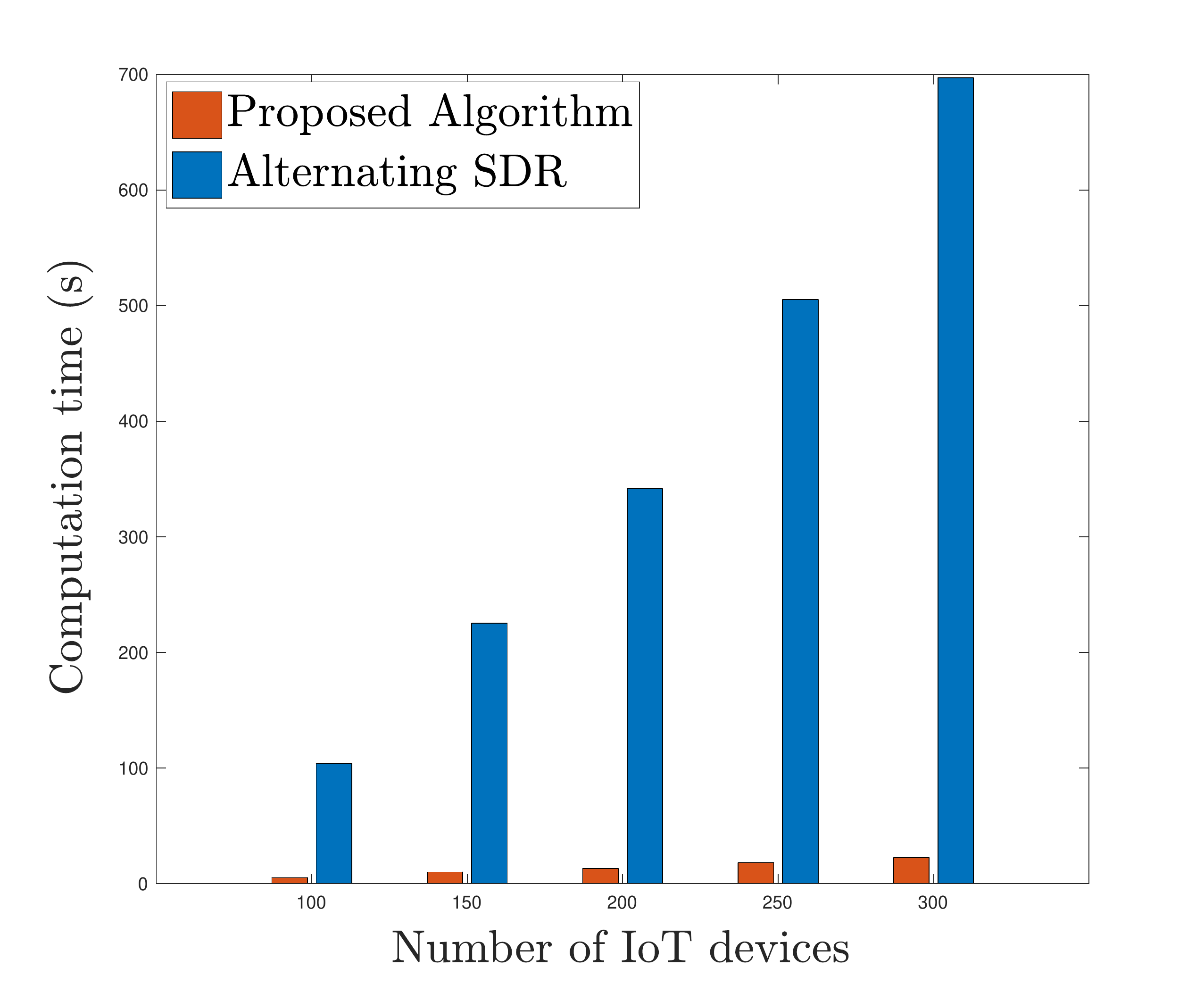}
                \vspace{-3mm}
                \caption{Computation time versus the number of IoT devices when $M=20$ and $N=50$.}\label{l:t:U}
                \vspace{-9mm}
    \end{figure}

The computation time of the considered algorithms versus the number of IoT devices is illustratd in Fig. \ref{l:t:U}. 
It is obvious that, compared to the benchmark, our proposed algorithm has a superior performance in terms of the computation time. 
In particular, the reduction ratios are approximately $97\%$ and $95\%$ when $K=300$ and $K=100$, respectively.
It can also be observed that the superiority of the proposed algorithm in terms of the computation time enlarges as the number of IoT devices increases. 
This indicates the potentials of our proposed algorithm in high-density RIS-assisted AirComp systems.

%In Fig. \ref{l:U}, we show the impact of the number of IoT devices on MSE when $K=200,~M=20,$ and $N=20$.
%As can be observed, the MSE slightly increases as the number of IoT devices in the AirComp system increases.
%Since the MSE is determined by the worst channel condition among the IoT devices to the AP.
%As a result, the more IoT devices in the AirComp system tends to lead to a larger MSE in a high probability. 
%In addition, as can be observed from Fig. \ref{l:t:U}, as the number of IoT devices increases, the time consumption of the alternating SDR method significantly increases while the proposed algorithm is scalable, which verifies the high time efficiency of our proposal.   
%Therefore, we can claim that our proposed algorithm attains the same performance as that of the state-of-the-art method but at significantly lower complexity.

\subsubsection{Low-density Scenario}
\begin{figure}[t]
        \centering
        \begin{minipage}{.46\textwidth}
                \centering
                \includegraphics[width=7.5cm,height=6cm]{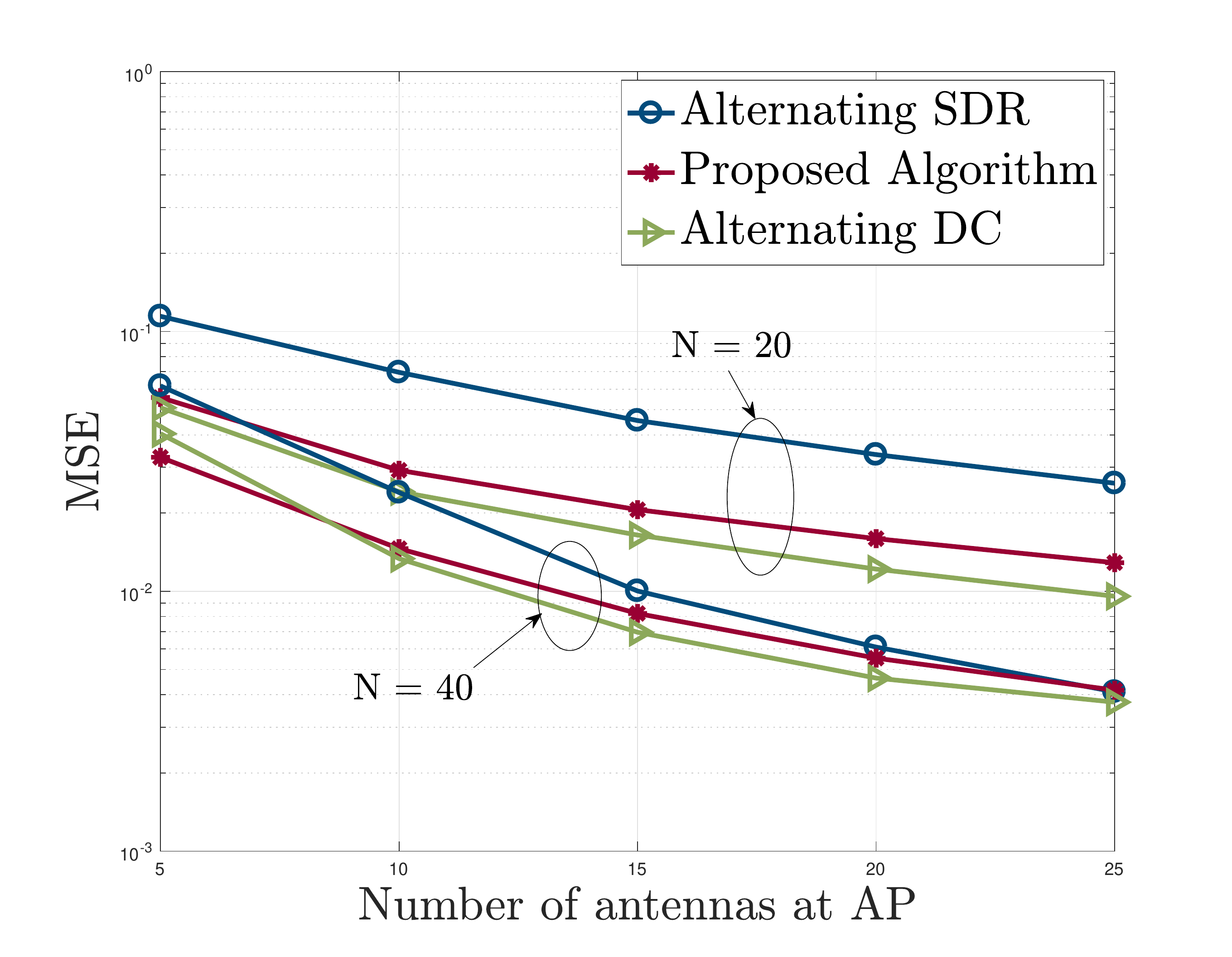}
                \vspace{-9mm}
                \caption{MSE versus the number of antennas at AP when $K=10$.}\label{DC0}
                \vspace{-6mm}
        \end{minipage}
        \begin{minipage}{.46\textwidth}
                \centering
                \includegraphics[width=7.5cm,height=6cm]{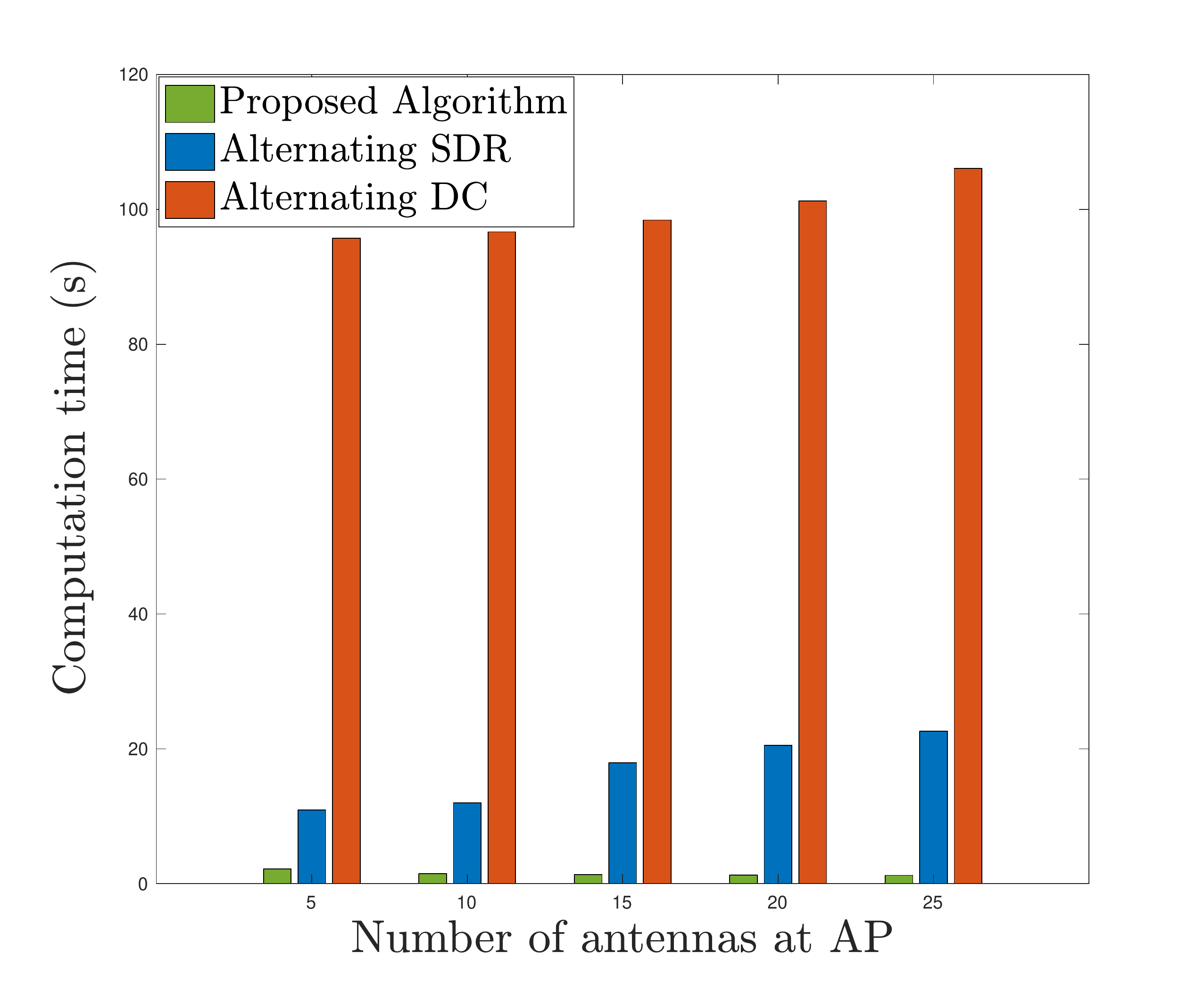}
                \vspace{-9mm}
                \caption{Computation time versus the number of antennas at AP when $K=10$ and $N=40$.}\label{DC1}
                \vspace{-6mm}
        \end{minipage}
\end{figure}

Due to the high computational complexity of the alternating DC algorithm, we compare its performance with the proposed algorithm when $K = 10$ in terms of the MSE and the computation time in Figs. \ref{DC0} and \ref{DC1}, respectively. 
It can be observed from Fig. \ref{DC0} that the proposed algorithm achieves a lower MSE than the alternating SDR algorithm and attains a slightly higher than MSE than the alternating DC algorithm. 
In addition, the MSE performance of all algorithms decreases as the number of reflecting elements at the RIS increases from 20 to 40. 
However, as shown in Fig. \ref{DC1}, our proposed algorithm is remarkably better than the alternating SDR algorithm and the alternating DC algorithm in terms of the computation time.
%The results suggest that the proposed algorithm attains a great trade-off between performance and complexity.

% or


\section{Conclusion}

In this paper, we proposed to leverage the advantage of RIS to mitigate the performance bottleneck of AirComp, thereby, achieving fast wireless data aggregation in IoT networks. 
We formulated an MSE minimization problem that requires the joint optimization of the transmit scalars at the IoT devices, the receive beamforming vector and the denoising factor at the AP, and the phase-shift matrix at the RIS. 
To solve this problem, a novel alternating minimization method in conjunction with the SCA technique was thus developed with convergence guarantee. 
To further reduce the computational complexity, we proposed a Mirror-Prox method that only involves a series of closed-form updates to solve the convex but non-smooth subproblem in each alternation. 
Simulations showed that, compared to the existing alternating SDR and alternating DC algorithms, the proposed algorithm can significantly reduce the computation time while achieving a similar MSE performance. 

\section*{Appendix}
\subsection{Proof of Proposition \ref{equi}}\label{}

For simplicity of notations, we denote 
\[
F_k(\bm m, \bm v) = \frac{\|\bm m\|^2}{\|\bm m^{\sf H}(\bm  h_{\mathrm{d},k}+\bm G \text{diag}(\bm  h_{\mathrm{r},k})\bm v )\|^2}, ~ \forall ~ k,
\]
such that problem \eqref{eq:p0} can be rewritten as
\begin{equation}
        \begin{split}
                \underset{\bm m,\bm v}{\operatorname{min}} ~& \max_k F_k(\bm m, \bm v) \\
                \text { s.t. } ~ & |v_i| = 1, \forall ~  i. %n=1,\cdots,N.
        \end{split}
\end{equation}
We then reformulate the problem into the following equivalent form:
\[
\begin{aligned}
        \begin{split}
                        \underset{\bm m,\bm v}{\operatorname{min}} ~& \max_k F_k(\bm m, \bm v) \\
            \text { s.t. } ~& |v_i| = 1, \forall ~  i,%n=1,\cdots,N.
                \end{split} 
                \quad \Longleftrightarrow \quad &
         \begin{split}
                         \underset{\bm m,\bm v}{\operatorname{max}} ~& \left( \max_k  F_k(\bm m, \bm v) \right)^{-1}\\
                \text { s.t. } ~& |v_i| = 1, \forall ~  i,%n=1,\cdots,N.
                \end{split}      \\
       \quad \Longleftrightarrow \quad &
        \begin{split}
            \underset{\bm m,\bm v}{\operatorname{max}} ~& \min_k F_k^{-1}(\bm m, \bm v) \\
            \text { s.t. } ~& |v_i| = 1, \forall ~  i, %n=1,\cdots,N.
        \end{split}
\end{aligned}
\]
where $F_k^{-1}(\bm m, \bm v) = \frac{1}{F_k(\bm m, \bm v)}$.
%\begin{equation}
%        \begin{split}
%                &\underset{\bm m,\bm v}{\operatorname{min}} \left(\max_k \frac{\|\bm m\|^2}{\|\bm m^{\sf H}(\bm  h_{\mathrm{d},k}+\bm G \text{diag}(\bm  h_{\mathrm{r},k})\bm v )\|^2}  \right)\\
%                &\text { s.t. } \quad |v_n| = 1, \forall ~  n=1,\cdots,N.
%        \end{split} \nonumber \Longleftrightarrow
%\end{equation}
%\begin{equation}
%         \begin{split} 
%                &\underset{\bm m,\bm v}{\operatorname{max}} \left(\max_k \frac{\|\bm m\|^2}{\|\bm m^{\sf H}(\bm  h_{\mathrm{d},k}+\bm G \text{diag}(\bm  h_{\mathrm{r},k})\bm v )\|^2} \right)^{-1}\\
%                &\text { s.t. } \quad |v_n| = 1, \forall ~  n=1,\cdots,N.
%        \end{split} \nonumber \Longleftrightarrow
%\end{equation}
%\begin{equation}
%        \begin{split}\label{eq:new}
%                &\underset{\bm m,\bm v}{\operatorname{max}} \left(\min_k \frac{\|\bm m^{\sf H}(\bm  h_{\mathrm{d},k}+\bm G \text{diag}(\bm  h_{\mathrm{r},k})\bm v )\|^2}{\|\bm m\|^2} \right)\\
%                &\text { s.t. } \quad |v_n| = 1, \forall ~  n=1,\cdots,N.
%        \end{split}
%\end{equation}
As a result, problem \eqref{eq:p0} is equivalent to 
\begin{equation}
        \begin{split}\label{eq:new}
                \underset{\bm m,\bm v}{\operatorname{max}} ~& \left\{\min_k \frac{\|\bm m^{\sf H}(\bm  h_{\mathrm{d},k}+\bm G \text{diag}(\bm  h_{\mathrm{r},k})\bm v )\|^2}{\|\bm m\|^2} \right\}\\
                \text { s.t. } ~& |v_n| = 1, \forall ~  n=1,\ldots,N.
        \end{split}
\end{equation}
Besides, we introduce an auxiliary variable $\tau=\|\bm m\|^2$.
Problem \eqref{eq:new} can be rewritten as the following form 
\begin{equation}\label{eq:pf2}
        \begin{split}
                \underset{\bm m,\bm v, \tau}{\max} ~& \left\{\min_k\frac{\|\bm m^{\sf H}(\bm  h_{\mathrm{d},k}+\bm G \text{diag}(\bm  h_{\mathrm{r},k})\bm v)\|^2}{\tau}\right\}\\
                 \text { s.t. } ~& |v_n| = 1,\forall ~  n=1,\ldots,N,\\
                ~& \|\bm m\|^2 = \tau.
        \end{split}
\end{equation}
By denoting $\hat{\bm m} = \frac{\bm m}{\sqrt{\tau}}$, problem \eqref{eq:pf2} can be represented as 
 \begin{equation}\label{eq:pf3}
        \begin{split}
                \underset{\hat{\bm m},\bm v}{\max} ~& \left\{\min_k \|\hat{\bm m}^{\sf H}(\bm  h_{\mathrm{d},k}+\bm G \text{diag}(\bm  h_{\mathrm{r},k})\bm v)\|^2\right\}\\
                \text { s.t. } ~& |v_n| = 1,\forall ~  n=1,\ldots,N,\\
                ~& \|\hat{\bm m}\|^2 = 1.
        \end{split}
\end{equation}
We further transform problem \eqref{eq:pf3} to its equivalent problem in the min-max form as follows 
%Specifically, problem \eqref{re:sub2} has the following equivalent form
\[
\begin{aligned}
        \begin{split}
                \underset{\hat{\bm m},\bm v}{\max} ~& \left\{\min_k \|\hat{\bm m}^{\sf H}(\bm  h_{\mathrm{d},k}+\bm G \text{diag}(\bm  h_{\mathrm{r},k})\bm v)\|^2\right\}\\
                \text { s.t. } ~& |v_n| = 1,\forall ~  n=1,\ldots,N,\\
                ~& \|\hat{\bm m}\|^2 = 1.
        \end{split} 
        & \Longleftrightarrow &
         \begin{split}
                \underset{\hat{\bm m},\bm v}{\min} ~& \left\{-\min_k \|\hat{\bm m}^{\sf H}(\bm  h_{\mathrm{d},k}+\bm G \text{diag}(\bm  h_{\mathrm{r},k})\bm v)\|^2\right\}\\
                \text { s.t. } ~& |v_n| = 1,\forall ~  n=1,\ldots,N,\\
                ~& \|\hat{\bm m}\|^2 = 1.
        \end{split} \\
       & \Longleftrightarrow &
        \begin{split}
                \underset{\hat{\bm m},\bm v}{\min} ~&\max_k \left\{-\|\hat{\bm m}^{\sf H}(\bm  h_{\mathrm{d},k}+\bm G \text{diag}(\bm  h_{\mathrm{r},k})\bm v)\|^2\right\}\\
                \text { s.t. } ~& |v_n| = 1,\forall ~  n=1,\ldots,N,\\
                ~& \|\hat{\bm m}\|^2 = 1.
        \end{split}
\end{aligned}
\]

%\section{Proof of Proposition 2}
%\begin{proof}
%Assuming the optimal solution of problem \eqref{re:sub1:ineq} is $\bm m^*$ and it corresponds to the optimal value $Q^*$.
%And the optimal value  attained by \eqref{re:sub1} is $W^*$.
%Then we have 
%\begin{equation}\label{<}
%       Q^* \geq W^*,
%\end{equation}
%since the feasible region of \eqref{re:sub1:ineq} contain \eqref{re:sub1}.
%On the other hand, we take a Euclidean projection on $\bm m^*$ to get $\bm {\bar m}$ which satisfies $$\|\bm {\bar m}\| = 1.$$
%Then we have $$\min_k ~\|\bm {\bar m}^{\sf H}\bm{h}_{k}\|^2 \geq \min_k ~\|{(\bm m^*)}^{\sf H}\bm{h}_{k}\|^2 .$$ 
%Thus we can assert 
%\begin{equation}\label{>}
%       W^* \geq Q^*,
%\end{equation}
%Finally, we can conclude that $Q^* = W^*$ according to \eqref{<} and \eqref{>}. 
%Thus, Problem \eqref{re:sub1} is equivalent to problem \eqref{re:sub1:ineq}.
%This completes the proof of Proposition 2.
%\end{proof}

\subsection{Proof of Proposition \ref{con}}\label{}

We first prove property (i).
Considering the $l$-th alternating iteration, for a given $\tilde{\bm m}_{(l)}^{(0)}$, we denote the objective value of problem \eqref{v:real:quad} as $f^v_{(l)}(\tilde{\bm v})$.
We consider the SCA iteration that starts from $\tilde{\bm v}_{(l)}^{(0)}$.
By denoting
$\hat{f}^v_{(l)}(\tilde{\bm v},\tilde{\bm v}_{(l)}^{(n)})$ as the objective value of problem \eqref{saddle:ori2}, it satisfies
\begin{enumerate}
\item $\hat{f}^v_{(l)}(\tilde{\bm v}_{(l)}^{(n)},\tilde{\bm v}_{(l)}^{(n)}) = f^v_{(l)}(\tilde{\bm v}_{(l)}^{(n)})$,
\item $f^v_{(l)}(\tilde{\bm v}_{(l)}^{(n+1)}) \leq \hat{f}^v_{(l)}(\tilde{\bm v}_{(l)}^{(n+1)},\tilde{\bm v}_{(l)}^{(n)})$,
\item $\hat{f}^v_{(l)}(\tilde{\bm v}_{(l)}^{(n+1)},\tilde{\bm v}_{(l)}^{(n)}) \leq \hat{f}^v_{(l)}(\tilde{\bm v}_{(l)}^{(n)},\tilde{\bm v}_{(l)}^{(n)}).$
\end{enumerate}
Inequality 1) holds since $\hat{f}^v_{(l)}(\tilde{\bm v},\tilde{\bm v}_{(l)}^{(n)})$ is a linear approximation of $f^v_{(l)}(\tilde{\bm v})$ at point $\tilde{\bm v}_{(l)}^{(n)}$.
According to \eqref{first}, $\hat{f}^v_{(l)}(\tilde{\bm v},\tilde{\bm v}_{(l)}^{(n)})$ is an upper bound of $f^v_{(l)}(\tilde{\bm v})$.
As a result, we can establish inequality 2).
Besides, inequality 3) holds as $\tilde{\bm v}_{(l)}^{(n+1)} = \argmin_{\tilde{\bm v} \in \mathcal{V}} \hat{f}^v_{(l)}(\tilde{\bm v},\tilde{\bm v}_{(l)}^{(n)})$.
Hence, we obtain the following chain inequalities
\[
f^v_{(l)}(\tilde{\bm v}_{(l)}^{(n+1)}) \leq \hat{f}^v_{(l)}(\tilde{\bm v}_{(l)}^{(n+1)},\tilde{\bm v}_{(l)}^{(n)}) \leq \hat{f}^v_{(l)}(\tilde{\bm v}_{(l)}^{(n)},\tilde{\bm v}_{(l)}^{(n)}) = f^v_{(l)}(\tilde{\bm v}_{(l)}^{(n)}) .
\]              
Besides, the continuous function $f^v_{(l)}(\tilde{\bm v})$ has a lower bound in the constrained set.
As a result, the non-increasing and lower bounded sequence $\{f^v_{(l)}(\tilde{\bm v}_{(l)}^{(n)})\}$ converges.
Similarly, we can establish the non-increasing and convergent property for the objective value sequence achieved by $\{\tilde{m}_{(l)}^{(n)}\}_{n=0}^{\infty}$.
To this end, we have proved property i).

On the other hand, we can prove property ii) by iteratively utilizing property i).

\subsection{Proof of Lemma \ref{P_D}}\label{}

        We have shown that problem \eqref{saddle:ori2} is equivalent to problem \eqref{socp}. 
        We then rewrite problem \eqref{socp} in a more compact form as
\begin{equation}\label{socp_}
        \begin{split}
                \underset{\tilde{\bm v}, t}{\min}  ~~& t \\
                \text { s.t. }  ~&  \mathbf{P}^{(n)} \tilde{\bm v} + \mathbf{q}^{(n)} - t\mathbf{1}  \preceq 0, \\
                                ~& \tilde{\bm v} \in \mathcal{V}, ~ t \in \mathbb{R}.
       \end{split}
\end{equation}
where $\mathbf{1} = [1,1,\ldots,1]^{\sf T} \in \mathbb{R}^{K}.$ 
For a given vector $\bm{y} \in \mathbb{R}^{K}$ with non-negative components, which is known as dual variable, the corresponding Lagrangian relaxed problem is given by  
\begin{equation}\label{socp_lag}
        \begin{split}
                \underset{\tilde{\bm v}, t}{\min}  ~~& t + \bm{y}^{\sf T}\left(\mathbf{P}^{(n)} \tilde{\bm v} + \mathbf{q}^{(n)} - t\mathbf{1} \right) \\
                \text { s.t. } ~& \tilde{\bm v} \in \mathcal{V}, ~ t \in \mathbb{R}.
       \end{split}
\end{equation}
We reorganize the objective function of problem \eqref{socp_lag} as 
\[
\left(1-\mathbf{1}^{\sf T} \bm{y}\right)t + \left(\mathbf{P}^{(n)} \tilde{\bm v} + \mathbf{q}^{(n)}\right)^{\sf T} \bm{y}.
\]
The dual objective function of problem \eqref{socp_}  is given by the optimal value of problem \eqref{socp_lag}.
If $1-\mathbf{1}^{\sf T} \bm{y} \neq 0$, then $\left(1-\mathbf{1}^{\sf T} \bm{y}\right)t$ is not bounded from below, since $\left(1-\mathbf{1}^{\sf T} \bm{y}\right)t \rightarrow -\infty$ if $\mathbf{1}^{\sf T} \bm{y}>1$ and $t \rightarrow+\infty$ or if $\mathbf{1}^{\sf T} \mathbf{y}<1$ and $t \rightarrow-\infty$.

Thus, we only need to consider $\bm{y} \in \mathbb{R}^{K}$ that satisfy $\bm{y} \geq \bm{0}$ and $\mathbf{1}^{\sf T} \bm{y}=1$. For these $\bm{y}$, the dual objective function is given by $\min _{\tilde{\bm v} \in \mathcal{V}}\left(\mathbf{P}^{(n)} \tilde{\bm v} + \mathbf{q}^{(n)}\right)^{\sf T} \bm y$.
The dual problem is 
\begin{equation}\label{derive:dual}
        \begin{split}
              \max_{\mathbf{y}} &~ \underset{\tilde{\bm v}}{\min} ~ \left(\mathbf{P}^{(n)} \tilde{\bm v} + \mathbf{q}^{(n)}\right)^{\sf T} \bm y\\
                \text { s.t. } &~  \tilde{\bm v} \in \mathcal{V}, ~ \bm y \in \mathcal{Y}.
        \end{split}
\end{equation}
Because the original problem is convex and set $\mathcal{V}$ is closed and compact, the strong duality condition is satisfied. 
Hence, its dual problem is equivalent to itself. 
As the objective function $\left(\mathbf{P}^{(n)} \tilde{\bm v} + \mathbf{q}^{(n)}\right)^{\sf T} \bm y$ is linear to $\tilde{\bm v}$ and $\bm y$, according to \cite{boyd_vandenberghe_2004}, problem \eqref{derive:dual} is equivalent to problem \eqref{ori:saddle4}, which is a smooth convex-concave saddle point problem \cite{variational}.

\subsection{Proof of Lemma \ref{Lipschitz}}
$\mathbf{F}(\cdot)$ is a linear operator with well-defined monotonicity. 
For the $L$-Lipschitz property of $\mathbf{F}(\cdot)$, we need verify the following inequality: 
\begin{equation}
\begin{array}{l}
\text{a).}\left\|\nabla_{\tilde{\bm v}} \psi^{(n)}(\tilde{\bm v}, \bm y)-\nabla_{\tilde{\bm v}} \psi^{(n)}\left(\tilde{\bm v}^{\prime}, \bm y\right)\right\| \leq 0\cdot\left\|\tilde{\bm v}-\tilde{\bm v}^{\prime}\right\|, \\
\text{b).}\left\|\nabla_{\tilde{\bm v}} \psi^{(n)}(\tilde{\bm v}, \bm y)-\nabla_{\tilde{\bm v}} \psi^{(n)}\left(\tilde{\bm v}, \bm y^{\prime}\right)\right\| \leq L \cdot\left\|\bm y-\bm y^{\prime}\right\|_{1}, \\
\text{c).}\left\|\nabla_{\bm y} \psi^{(n)}(\tilde{\bm v}, \bm y)-\nabla_{\bm y} \psi^{(n)}\left(\tilde{\bm v}, \bm y^{\prime}\right)\right\|_{\infty} \leq 0\cdot\left\|\bm y-\bm y^{\prime}\right\|_{1}, \\
\text{d).}\left\|\nabla_{\bm y} \psi^{(n)}(\tilde{\bm v}, \bm y)-\nabla_{\bm y} \psi^{(n)}\left(\tilde{\bm v}^{\prime}, \bm y\right)\right\|_{\infty} \leq L\cdot\left\|\tilde{\bm v}-\tilde{\bm v}^{\prime}\right\|,
\end{array} \nonumber
\end{equation}
where $L = \max_{k}\{\|\mathbf{p}_{k}^{(n)}\|\}$, $\|\cdot\|$ and $\|\cdot\|_{\infty}$ are the dual norm of $\|\cdot\|$ and $\|\cdot\|_{1}$, respectively. 
First, inequalities a) and c) hold due to the fact that
\[\nabla_{\tilde{\bm v}} \psi^{(n)}(\tilde{\bm v}, \bm y) = \nabla_{\tilde{\bm v}} \psi^{(n)}\left(\tilde{\bm v}^{\prime}, \bm y\right)=(\mathbf{P}^{(n)})^{\sf T}\bm y, \]
\[\nabla_{\bm y} \psi^{(n)}(\tilde{\bm v}, \bm y) = \nabla_{\bm y} \psi^{(n)}\left(\tilde{\bm v}, \bm y^{\prime}\right) = \mathbf{P}^{(n)} \tilde{\bm v} + \mathbf{q}^{(n)} .\]
Inequality b) holds since  
\begin{equation}
\begin{aligned}
&\quad\left\|\nabla_{\tilde{\bm v}} \psi^{(n)}(\tilde{\bm v}, \bm y)-\nabla_{\tilde{\bm v}} \psi^{(n)}\left(\tilde{\bm v}, \bm y^{\prime}\right)\right\| \\  
&= \|(\mathbf{P}^{(n)})^{\sf T}\bm y - (\mathbf{P}^{(n)})^{\sf T}\bm y^{\prime}\| = \|\sum_{k=1}^{K}\mathbf{p}_{k}^{(n)}(y_k-y_k^{\prime})\| \\ 
&\leq \sum_{k=1}^{K}\|\mathbf{p}_{k}^{(n)}\| \cdot |y_k-y_k^{\prime}| \leq  \left(\max_{k}\|\mathbf{p}_{k}^{(n)}\|\right) \cdot \left(\sum_{k=1}^{K}|y_k-y_k^{\prime}|\right) \\
&= L \cdot \|\bm y - \bm y^{\prime}\|_1.
\end{aligned}
\nonumber
\end{equation}
Finally, inequality d) holds as  
\begin{equation}
\begin{aligned}
&\quad\left\|\nabla_{\bm y} \psi^{(n)}(\tilde{\bm v}, \bm y)-\nabla_{\bm y} \psi^{(n)}\left(\tilde{\bm v}, \bm y^{\prime}\right)\right\|_{\infty} \\  
&= \|\mathbf{P}^{(n)}\tilde{\bm v} - \mathbf{P}^{(n)}\tilde{\bm v}^{\prime}\|_{\infty} =\max_{k}\left\{\left|\left(\mathbf{p}_{k}^{(n)}\right)^{\sf T}(\tilde{\bm v} - \tilde{\bm v}^{\prime})\right|\right\} \\ 
&\underset{(a)} \leq \left(\max_{k}\|\mathbf{p}_{k}^{(n)}\|\right) \cdot \| \tilde{\bm v} - \tilde{\bm v}^{\prime} \| \\
&= L \cdot \|\tilde{\bm v} - \tilde{\bm v}^{\prime}\|,
\end{aligned}
\nonumber
\end{equation}
where $(a)$ follows by applying the Cauchy-Schwarz inequality.

\bibliographystyle{IEEEtran}  % set style to IEEE
\bibliography{reference.bib}

\ifCLASSOPTIONcaptionsoff
  \newpage
\fi

\end{document}